\documentclass[11pt,reqno]{amsart}%
\usepackage[a4paper,margin=3cm]{geometry}%
\usepackage[pagebackref]{hyperref}%
\usepackage[T1]{fontenc}%
\usepackage[utf8]{inputenc}%
\usepackage{lmodern,mathtools}

\title{Dynamics of a planar Coulomb gas}

\author{François Bolley}
\address[FB]{LPSM, CNRS UMR 8001, Sorbonne Université - Paris 6, France.}
\email{\url{mailto:francois.bolley@upmc.fr}}
\urladdr{\url{http://www.proba.jussieu.fr/pageperso/bolley/}}

\author{Djalil Chafaï} 
\address[DC]{CEREMADE, CNRS UMR 7534, Université Paris-Dauphine, PSL, France.}
\email{\url{mailto:djalil(at)chafai.net}} %
\urladdr{\url{http://djalil.chafai.net/}}

\author{Joaquín Fontbona}
\address[JF]{CMM, Universidad de Chile, Chile.}
\email{\url{mailto:fontbona@dim.uchile.cl}}
\urladdr{\url{http://www.cmm.uchile.cl/?cmm_people=joaquin-fontbona}}

\date{Summer 2017, revised Winter 2018, compiled \today}

\newtheorem{theorem}{Theorem}[section]%
\newtheorem{proposition}[theorem]{Proposition}%
\newtheorem{lemma}[theorem]{Lemma}%
\newtheorem{remark}[theorem]{Remark}%

\newcommand{\dC}{\mathbb{C}}
\newcommand{\dE}{\mathbb{E}}

\newcommand{\dP}{\mathbb{P}}
\newcommand{\dR}{\mathbb{R}}

\newcommand{\cC}{\mathcal{C}}
\newcommand{\cE}{\mathcal{E}}\newcommand{\cF}{\mathcal{F}}

\newcommand{\cM}{\mathcal{M}}\newcommand{\cN}{\mathcal{N}}
\newcommand{\cP}{\mathcal{P}}

\newcommand{\cS}{\mathcal{S}}

\newcommand{\al}{\alpha}
\newcommand{\be}{\beta}
\newcommand{\De}{\Delta}
\newcommand{\de}{\delta}
\newcommand{\ga}{\gamma}
\newcommand{\Ga}{\Gamma}
\newcommand{\la}{\lambda}

\newcommand{\na}{\nabla}

\newcommand{\si}{\sigma}

\newcommand{\te}{\theta}

\newcommand{\veps}{\varepsilon}
\newcommand{\vphi}{\varphi}

\newcommand{\ABS}[1]{{{\left| #1 \right|}}} 
\newcommand{\DOTP}[2]{#1\cdot#2} 
\newcommand{\PAR}[1]{{{\left(#1\right)}}} 
\newcommand{\pd}{{\partial}} 
\newcommand{\SBRA}[1]{{{\left[#1\right]}}} 

\newcommand{\IND}{\mathbf{1}}

\newcommand{\OL}[1]{\overline{#1}}

\newcommand{\Wt}{\widetilde}

\newcommand{\SSK}[1]{\substack{#1}}

\renewcommand{\Re}{\mathfrak{Re}}

\newcommand{\dd}{\mathrm{d}}
\newcommand{\e}{\mathrm{e}}

\makeatletter
\def\@MRExtract#1 #2!{#1}     
\renewcommand{\MR}[1]{
  \xdef\@MRSTRIP{\@MRExtract#1 !}%
  \href{http://www.ams.org/mathscinet-getitem?mr=\@MRSTRIP}{MR-\@MRSTRIP}}
\makeatother

\numberwithin{equation}{section}



\begin{document}

\begin{abstract}
  We study the long-time behavior of the dynamics of interacting planar
  Brownian particles, confined by an external field and subject to a singular
  pair repulsion. The invariant law is an exchangeable Boltzmann\,--\,Gibbs
  measure. For a special inverse temperature, it matches the Coulomb gas known
  as the complex Ginibre ensemble. The difficulty comes from the interaction
  which is not convex, in contrast with the case of one-dimensional log-gases
  associated with the Dyson Brownian Motion. Despite the fact that the
  invariant law is neither product nor log-concave, we show that the system is
  well-posed for any inverse temperature and that Poincaré inequalities are
  available. Moreover the second moment dynamics turns out to be a nice
  Cox\,--\,Ingersoll\,--\,Ross process, in which the dependency over the number
  of particles leads to identify two natural regimes related to the behavior
  of the noise and the speed of the
  dynamics. 
\end{abstract}

\maketitle
 
{\footnotesize\tableofcontents}

\section{Introduction and statement of the results}

\subsection{The model and its well-posedness}

This work is concerned with the dynamics of $N\geq2$ particles at positions
$x_1,\ldots,x_N$ in $\dR^d$, $d\geq1$, confined by an external field and
experiencing a singular pair repulsion. The configuration space that we are
interested in is the open subset $D\subset(\dR^d)^N$ defined by
\begin{equation}\label{eq:D}
  D:=(\dR^d)^N\setminus\cup_{i\neq j}\{(x_1,\ldots,x_N)\in(\dR^d)^N:x_i=x_j\}
\end{equation}
where $i,j$ run over $\{1,\ldots,N\}$. The boundary of $D$ in the
compactification of $(\dR^d)^N$ is
\[
\pd D:=\{\infty\}\cup\cup_{i\neq j}\{(x_1,\ldots,x_N)\in(\dR^d)^N:x_i=x_j\}.
\]
The vector $x=(x_1,\ldots,x_N)\in D$ encodes the position of the $N$ particles,
and the energy $H(x)$ of this configuration  is
modeled by
\begin{equation}\label{eq:E}
  H(x):=\frac{1}{N}\sum_{i=1}^N V(x_i) %
  +\frac{1}{2N^2}\sum_{1 \leq i\neq j \leq N}W(x_i-x_j)
  =:{H_V}(x)+{H_W}(x). 
\end{equation}
Here, $V:\dR^d\to\dR\cup\{+\infty\}$ is an external \emph{confinement
  potential} such that $V(z)\to+\infty$ as $z\to\infty$, and
$W:\dR^d\setminus\{0\}\to\dR$ is a \emph{pair or two-body interaction
  potential} such that $W(z)=W(-z)$ and $W(z)\to+\infty$ as $z\to0$
(singularity). Unless otherwise stated, we consider particles in
$\dR^2\equiv\dC$, with quadratic confinement and Coulomb repulsion, namely:
\begin{equation}\label{eq:Ginibre}
  d=2,
  \quad
  V(z)=\ABS{z}^2,
  \quad
  W(z)=\log\frac{1}{\ABS{z}^2}.
\end{equation}
Here $\ABS{z}$ denotes the Euclidean norm of $z\in\dR^2$ (modulus of the
complex number $z$). With this notation, we study the system of $N$
interacting particles in $\dR^2$ modeled by a diffusion process
$X^N={(X^N_t)}_{t\geq0}$ on $D$, solution of the stochastic differential
equation
\begin{equation}\label{eq:SDE}
  \dd X^N_t=\sqrt{2\frac{\al_N}{\be_N}}\dd B^N_t-\al_N\na H(X^N_t)\,\dd t,
\end{equation}
for any choice of speed $\al_N>0$ and inverse temperature $\be_N>0$; here
${(B^N_t)}_{t\geq0}$ is a standard Brownian motion of $(\dR^2)^N$. In other
words, letting $X^N_t={(X^{i,N}_t)}_{1\leq i\leq N}$ and
$B^N_t={(B^{i,N}_t)}_{1\leq i\leq N}$ denote the components of $X^N_t$ and
$B^N_t$,
\[
\dd X^{i,N}_t=\sqrt{2\frac{\al_N}{\be_N}}\dd B^{i,N}_t 
-\frac{\al_N}{N}\na V(X^{i,N}_t)\,\dd t
-\frac{\al_N}{N^2}\sum_{j\neq i}\na W(X^{i,N}_t-X^{j,N}_t)\,\dd t,
\quad 1\leq i\leq N.
\]
Since $V(z)=\ABS{z}^2$ and $W(z)=-2\log\ABS{z}$ we have more explicitly
\begin{equation}\label{processX}
\dd X^{i,N}_t
=\sqrt{2\frac{\al_N}{\be_N}}\dd B^{i,N}_t 
-2\frac{\al_N}{N}X^{i,N}_t\,\dd t
-2\frac{\al_N}{N^2}\sum_{j\neq i}\frac{X^{j,N}_t-X^{i,N}_t}{|X^{i,N}_t-X^{j,N}_t|^2}\,\dd t,
\quad 1\leq i\leq N.
\end{equation}
To lightweight the notations, we will very often drop the notation $N$ in the
superscript, writing in particular $X_t$, $B_t$, $X^i_t$, and $B^i_t$ instead
of $X^N_t$, $B^N_t$, $X^{i,N}_t$ and $B^{i,N}_t$ respectively. We shall see
later that the cases $\be_N=N$ and $\be_N=N^2$ are particularly interesting,
the latter being related to the complex Ginibre Ensemble in random matrix
theory.

Global pathwise well posedness of a solution $X$ to the stochastic
differential equation \eqref{processX} is not automatically granted since $W$
is singular. Nevertheless, the set $D$ is path-connected (see Lemma
\ref{le:connectivity}) and, given an initial condition $X_0$ in $D$, one can
resort to classic stochastic differential equations properties to define, in a
unique pathwise way, the process $X^N$ up to the explosion time
\begin{equation}\label{eq:Texp}
  T_{\pd D}:=\sup_{\veps>0}T_\veps\in[0,+ \infty]. 
\end{equation}
Here, 
\begin{equation*}
  T_\veps
  =\inf\bigg\{t\geq0:\max_{1\leq i\leq N}|X^i_t|\geq\veps^{-1}
  \text{ or }
  \min_{1\leq i\leq N}|X^i_t-X^j_t|\leq\veps
  \bigg\}
\end{equation*}
is the first exit time of a typical compact set in $D$. Then, one can show
that explosion never occurs:

\begin{theorem}[Global well posedness and absence of explosion]\label{th:boum}
  For any $X_0=x\in D$, pathwise uniqueness and strong existence on
  $[0, + \infty)$ hold for the stochastic differential equation \eqref{processX}
  on $[0,+\infty)$, and we have $T_{\pd D}= + \infty$ a.s.
\end{theorem}

The absence of explosion provided by Theorem \ref{th:boum} is remarkably
independent of the choice of the inverse temperature, and this is in contrast
with the behavior of the Dyson Brownian motion associated with the
one-dimensional log-gas, see for instance \cite{MR1217451}. The proof of
Theorem \ref{th:boum} is given in Section \ref{se:boum}. It uses the fact that
$W$ is the fundamental solution of the Poisson\,--\,Laplace equation. The main
idea is similar to the one used for other singular repulsion models, such as
in \cite{MR1217451}, or for vortices such as in \cite{MR2343519}, but the
result ultimately relies on quite specific properties of our model
\eqref{processX}. Note also that our particles will never collide and in
particular never collide at the same time, in contrast with for instance the
singular attractive model studied in \cite{fournier-jourdain} -- see also
\cite{cattiaux-pedeches} for the control of explosion using the Fukushima
technology.

\medskip
  
Hence there exists a unique Markov process $X={(X_t)}_{t\geq0}$ solution of
\eqref{eq:SDE}. Its infinitesimal generator $L$ is given for a smooth enough
$f:D\to\dR$ by
\begin{equation}\label{eq:gen}
  L f=\frac{\al_N}{\be_N}\De f-\al_N\DOTP{\na H}{\na f}.
\end{equation}
Here $\De$ and $\na$ are understood in $(\dR^2)^N\equiv\dR^{2N}$ and $u\cdot
v=\langle u,v\rangle$ denotes the Euclidean scalar product. By symmetry of the
evolution, the law of $X_t$ is exchangeable for every $t\geq0$, as soon as it
is exchangeable for $t=0$. Recall that the law of a random vector is
exchangeable when it is invariant by any permutation of the coordinates of the
vector. It is then natural to encode the particle system with its empirical
measure
\[
\mu^N_t=\frac{1}{N}\sum_{i=1}^N\delta_{X^i_t}.
\]

\subsection{Second moment dynamics}

Theorem \ref{th:2ndmoment} gives the evolution of the second moment
\[
H_V(X_t)
=\frac{1}{N}\sum_{i=1}^N|X_t^i|^2
=\int_{\dR^2}\!\ABS{x}^2\,\mu^N_t(\mathrm{d}x)
\]
of $\mu^N_t$. This evolution depends on the choices for $\al_N$ and $\be_N$,
for which meaningful choices are discussed in Section \ref{se:comments}. We
let $W_1$ denote the (Kantorovich\,--)\,Wasserstein transportation distance of
order one defined by $W_1(\mu,\nu)=\inf\ \{\dE[|X-Y|]:X\sim\mu,Y\sim~\nu\}$
for every probability measures $\mu$ and $\nu$ on $\dR$ with finite first
moment.

\begin{theorem}[Second moment dynamics]\label{th:2ndmoment}
  The process ${(H_V(X_t))}_{t\geq0}$ is an ergodic Markov process, equal in
  law to the Cox\,--\,Ingersoll\,--\,Ross process ${(R_t)}_{t\geq0}$ given by the
  unique solution in $[0,\infty)$ of the stochastic differential equation
  \begin{equation}\label{SDECIR}
    \dd R_t =  
    \sqrt{\frac{8 \alpha_N}{ N \beta_N} R_t} \, \dd b_t 
    +4\frac{\al_N}{N} \left[\frac{ N}{\beta_N}+ \frac{N-1}{2N}-R_t \right]\dd t,
  \end{equation}
  where ${(b_t)}_{t\geq0}$ is a real standard Brownian motion. In particular,
  its invariant distribution is the Gamma law $\Ga_N$ on $\dR_+$ with shape
  parameter $N+\frac{N-1}{2N}\be_N$ and scale parameter $\be_N,$ and density
  with respect to the Lebesgue measure on $\dR$ given by
  \[
  r\in\dR\mapsto 
  \gamma_N(r):=
  \frac{\beta_N^{N+\frac{\beta_N(N-1)}{2 N}}}{\Gamma(N+\frac{\beta_N(N-1)}{2 N} )} 
  r^{(N-1)(1+\frac{\beta_N}{2 N})}\e^{-r\beta_N}\mathbf{1}_{r\geq0}.
  \]
  Moreover, for any $t \geq 0$ we have
  \begin{equation}\label{eq:cvCIR}
    W_1(\mathrm{Law}(H_V(X_t)),\Ga_N)
    \leq \e^{-  4 \frac{\al_N}{N}  t}\  W_1(\mathrm{Law}(H_V(X_0)), \Ga_N).
  \end{equation}
  Furthermore for any $x\in D$ and $t\geq0$, we have
  \begin{equation}\label{eq:mom2evol}
    \dE[H_V(X_t)\mid X_0=x]\\
    ={H_V}(x)\,\e^{-4\al_Nt/N} %
    +\PAR{\frac{1}{2}
      +\frac{N}{\be_N}-\frac{1}{2N}}\PAR{1-\e^{-4\al_Nt/N}}.
  \end{equation}
  In particular, as $t\to\infty$, the left-hand sides in~\eqref{eq:cvCIR}
  and~\eqref{eq:mom2evol} converge to $0$ and $1/2+N/\be_N-1/(2N)$
  respectively with a speed independent of $N$ as soon as $\al_N$ is linear in
  $N$.
\end{theorem}

A Cox\,--\,Ingersoll\,--\,Ross (CIR) process also naturally arises as the
dynamics of the second empirical moment of the vortex system studied in
\cite{MR3308573}. Theorem \ref{th:2ndmoment} is proved in
Section~\ref{se:2ndmoment}.

\subsection{Invariant probability measure and long-time behavior}

Despite the repulsive interaction, the confinement is strong enough to give
rise to an equilibrium. Namely, the Markov process ${(X_t)}_{t\geq0}$ admits a
unique invariant probability measure which is reversible. It is the
Boltzmann\,--\,Gibbs measure $P^N$ on $D\subset (\dR^2)^N$ with density
\begin{equation}\label{eq:PN}
  \frac{\dd P^N(x_1,\ldots,x_N)}{\dd x_1\cdots\dd x_N}
  =\frac{\e^{-\be_NH(x_1,\ldots,x_N)}}{Z_N}  
  =\frac{\e^{-\frac{\be_N}{N}\sum_{i=1}^N \vert x_i \vert^2}}{Z_N}
  \prod_{1 \leq i<j \leq N}\ABS{x_i-x_j}^{\frac{2\be_N}{N^2}}
\end{equation}
where 
\[
Z_N:=\int_D\!\e^{-\be_NH(x_1,\ldots,x_N)}\dd x_1\cdots\dd x_N
\]
is a normalizing constant known as the partition function. Such a
Boltzmann\,--\,Gibbs measure with a Coulomb interaction is called a Coulomb
gas. Actually $H(x)\to+\infty$ when $x\to\partial D$ and $\e^{-\be H}$ is
Lebesgue integrable on $D$ for any $\be>0$, see Lemma \ref{le:coercivity}.
Moreover the density of $P^N$ does not vanish on $D$. One can extend it on
$(\dR^2)^N$ by zero, seeing $P^N$ as a probability measure on $(\dR^2)^N$.
Since the domain $D$ and the function $H$ are both invariant by permutation of
the $N$ particles, the law $P^N$ is exchangeable. The behavior of $P^N$ relies
crucially on the ``inverse temperature'' $\be_N$. The choice $\be_N=N^2$ gives
a determinantal structure to $P^N$ which is known in this case as the complex
Ginibre ensemble in random matrix theory. As we will see in Section
\ref{ss:beta}, there is another interesting regime which is $\be_N=N$.

\bigskip

In Theorem \ref{th:Lyapunov} below we quantify the long time behavior of our
Markov process $X^N$ via a Poincaré inequality for its invariant measure
$P^N$. Recall that if $S$ is an open subset of $\dR^n$ and $\cF$ is a class of
smooth functions on $S$, then a probability measure $\mu$ on $S$ satisfies a
Poincaré inequality on $\cF$ with constant $c>0$ if for every $f\in\cF$,
\begin{equation}\label{eq:PI}
  \mathrm{Var}_\mu(f):=\dE_\mu(f^2)-(\dE_\mu f)^2
  \leq c\, \dE_\mu(\ABS{\nabla f}^2)
  \quad\text{where}\quad
  \dE_{\mu}(f):=\int\!f\,\dd\mu,
\end{equation}
see \cite{MR2352327} for instance. If $f$ is the density with respect to $\mu$
of a probability measure $\nu$ then the quantity
$\mathrm{Var}_\mu(f)=\dE_\mu(|f-1|^2)$ is nothing else but the chi-square
divergence $\chi(\nu\|\mu)$.

\begin{theorem}[Poincaré inequality]\label{th:Lyapunov}
  Let $\cF$ be the set of $\cC^\infty$ functions $f:D\to\dR$ with compact
  support in $D$, in the sense that the closure of $\{x\in D:f(x)\neq0\}$ is
  compact and is included in $D$. Then for any $N$, the probability measure
  $P^N$ on $(\dR^2)^N$ satisfies a Poincaré inequality on $\cF$ with a
  constant which may depend on $N$.
\end{theorem}

By \eqref{eq:gen}, the invariance of $P^N$ gives
\[
-\dE_{P^N}(fL f)=\frac{\al_N}{\be_N}\dE_{P^N}(\ABS{\na f}^2)
\]
where $\ABS{\na f}^2=\sum_{i=1}^n\ABS{\na_{x_i}f}^2$ for $\na
f=\PAR{\nabla_{x_i}f}_{1\leq i\leq N}$ in $(\dR^2)^N$. Let $P^N_t$ be the law
of $X^N_t$ in $(\dR^2)^N$. Up to determining a dense class of test functions
stable by the dynamics, it is classical, see \cite[Sec.~3.2]{MR2352327}
or~\cite{MR3155209}, that the Poincaré inequality \eqref{eq:PI} for $P^N$ with
constant $c=c_N$ imply the exponential convergence of $P^N_t$ to $P^N$, namely
\[
  \chi(P^N_t\| P^N)
  \leq\e^{-\frac{2t}{c_N}\frac{\al_N}{\be_N}}
  \chi(P^N_0\| P^N).
\]
More precisely, provided we already know that $P^N_t$ has a smooth density
$f^N_t$, we have
\[
\frac{\dd}{\dd t}\mathrm{Var}_{P^N}(f^N_t)
=\frac{\dd}{\dd t}\int\!(f^N_t)^2\,\dd P^N
=2\int\!f^N_tLf^N_t\,\dd P^N
\leq -2\frac{\al_N}{\be_Nc_N}\mathrm{Var}_{P^N}(f^N_t).
\]
Theorem \ref{th:Lyapunov} is proved in Section \ref{se:Lyapunov}. Poincaré
inequalities can classically be proved by spectral decomposition,
tensorization, convexity, perturbation, or Lipschitz deformation arguments,
see \cite{MR3155209}. None of these approaches seem to be available for $P^N$.

\begin{remark}[Eigenvector]\label{rm:eigenvector}
  It turns out that $H_V$ is up to an additive constant an eigenvector of $L$.
  Namely, from \eqref{eq:LHV} we get
  \[
    LU=-\frac{4\al_N}{N}U
    \quad\text{where}\quad
    U:=H_V-\frac{N}{\be_N}-\frac{N-1}{2N}.
  \]
  This fact is the key of the proof of Theorem \ref{th:2ndmoment}. However,
  due to the varying sign of $U$, we do not know how to use $U$ with the
  Lyapunov method to get a Poincaré inequality.
\end{remark}

\begin{remark}[Tensorization]
  The invariant measure $P^N$ of $X$ is not product, in contrast for instance
  with the case of vortex models with constant intensity studied in
  \cite{MR2343519}.
\end{remark}

\begin{remark}[Convexity]
  Neither the domain $D$ nor the energy $H:D\to\dR$ are convex, see
  Proposition \ref{pr:noconvex}, and thus the law $P^N$ is not log-concave.
  Remarkably, for one-dimensional log-gases, one can order the particles,
  which has the effect of producing a convex domain instead of $D$ on which
  $H$ is convex, and in this case $P^N$ satisfies in fact a logarithmic
  Sobolev inequality which is stronger, see for instance the forthcoming book
  \cite{book-erdos-yau} and also \cite{chafai-blog} for the optimal Poincaré
  constant. Here $d=2$ and the one dimensional trick is not available.
\end{remark}

\begin{remark}[Lipschitz deformation]
  The law $P^N$ is not a Lipschitz deformation of the Gaussian law on
  $\cM_N(\dC)$. Actually, the map which to $M\in\cM_N(\dC)$ associates its
  eigenvalues in $\dC$ is not Lipschitz. To see it take $M,M'\in\cM_n(\dC)$
  with $M_{j,j+1}=1$ for $j=1,\ldots,n-1$ and $M_{jk}=0$ otherwise, and
  $(M'-M)_{jk}=\veps$ if $(j,k)=(n,1)$ and $(M'-M)_{jk}=0$ otherwise. Then the
  eigenvalues of $M'-M$ are 
  \[
  \{\veps^{1/n}\e^{2ik\pi/n}:0\leq k\leq n-1\},
  \]
  while the Hilbert\,--\,Schmidt norm and operator norm of $M'-M$ are both
  equal to $\veps$. Note that in contrast, this map is Lipschitz for Hermitian
  matrices and more generally for normal matrices; this statement is known as
  the Hoffman\,--\,Wielandt inequality \cite{MR1288752}.
\end{remark}

The proof of Theorem \ref{th:Lyapunov} is based on a Lyapunov function and as
usual this does not provide in general a good dependence on $N$. Of course it
is natural to ask about the dependence in $N$ and in $\al_N$ and $\be_N$ of
the best constant in Theorem \ref{th:Lyapunov} and, specifically, if
convergence to equilibrium can be expected to hold at a rate that does not
depend on $N$, as in \cite{MR1847094}. Theorem \ref{th:Bobkov}
below and the previous Theorem \ref{th:2ndmoment} and Remark
\ref{rm:eigenvector} constitute steps in that direction.

\begin{theorem}[Uniform Poincaré inequality for the one particle
  marginal]\label{th:Bobkov}
  If $\be_N=N^2$ then the one-particle marginal law $P^{1,N}$ of $P^N$ on
  $\dR^2$ satisfies a Poincaré inequality, with a constant which does not
  depend on $N$. In particular, the smallest (i.e.\ best) constant for $P^N$
  is bounded below uniformly in $N$.
\end{theorem}

Theorem \ref{th:Bobkov} is proved in Section \ref{se:Bobkov}. 

\bigskip

Although the measure $P^N$ is not product, at least in the regime $\be_N=N^2$
a product structure arises asymptotically as $N$ goes to infinity. More
precisely, for $k \leq N$, let $P^{k,N}$ be the $k$-th dimensional marginal
distribution of the exchangeable probability measure $P^N$, as in
\eqref{eq:phikn}; then, in the regime $\be_N=N^2$, we have
 \begin{equation}\label{eq:decouplage}
   P^{k,N}-(P^{1,N})^{\otimes k} \to 0, \qquad N \to \infty
  \end{equation}
  weakly with respect to continuous bounded functions. It follows from
  Theorem~\ref{th:chaos} below.

\bigskip\bigskip

\begin{theorem}[Chaoticity]\label{th:chaos}
  Let $\beta_N = N^2$ and let $\mu_\infty$ be the uniform distribution on the
  unit disc $\{z\in\dC:|z|\leq1\}$ with density $\vphi_{\infty}(z) =
  \pi^{-1}\mathbf{1}_{\ABS{z}\leq1}$. For every fixed $k\geq1$,
  \[
  P^{k,N} \to \mu_\infty^{\otimes k}, \qquad N\to\infty
  \]
  weakly with respect to continuous and bounded functions. Moreover, denoting
  $\vphi^{k,N}$ the density of the marginal distribution $P^{k,N}$, as defined
  in \eqref{eq:phikn}, we have
  \[
  \vphi^{1,N} \to \vphi_{\infty}
  \quad \textrm{and} \quad
  \vphi^{2,N} \to \vphi_{\infty}^{\otimes 2}, \qquad   N\to\infty
  \]
  uniformly on compact subsets of   respectively
  \[
  \{z \in \dC: |z|\neq1\} \quad \textrm{and} \quad
  \{(z_1,z_2)\in\dC^2:|z_1|\neq1,|z_2|\neq1,z_1\neq z_2\}.
  \]
\end{theorem}

Theorem \ref{th:chaos} is proved in~Section \ref{se:chaos}. Note that the
convergence of $\vphi^{1,N}$ cannot hold uniformly on arbitrary compact sets
of $\dC$ since the pointwise limit is not continuous on the unit circle.
Moreover the convergence of $\vphi^{2,N}$ cannot hold on
$\{(z,z):z\in\dC, |z|<1\}$ since, by~\eqref{eq:phikn}, $\vphi^{2,N}(z,z)=0$
for any $N\geq2$ and $z\in\dC$ while
$\vphi^{1,N}(z)\vphi^{1,N}(z) \to_N 1/\pi^2\neq0$ when $|z|<1$, and this
phenomenon is due to the singularity of the interaction.

The case $\be_N=N^2$ is related to random matrix theory, see Section
\ref{ss:beta}. To our knowledge, Theorem \ref{th:Bobkov} and
Theorem~\ref{th:chaos} have not appeared previously in this domain.

\medskip

\subsection{Comments and open problems}\label{se:comments}

\subsubsection{Inverse temperature}\label{ss:beta}

Following \cite{MR3262506}, there are two natural regimes $\be_N=N$ and
$\be_N=N^2$.

\begin{itemize}
\item \textbf{Random matrix theory regime: $\be_N=N^2$.} This is natural from
  the point of view of random matrices. Namely let $M$ be a random $N\times N$
  complex matrix with independent and identically distributed Gaussian entries
  on $\dC$ with mean $0$ and variance $1/N$ with density
  $z\in\dC\mapsto\pi^{-1}N\exp(-N\ABS{z}^2)$. The variance scaling is chosen
  so that by the law of large numbers, asymptotically as $N\to\infty$, the
  rows and the columns of $M$ are stabilized: they have unit norm and are
  orthogonal. The density of the random matrix $M$ is proportional to
  \[
  M\mapsto
  \prod_{1\leq j,k\leq N}\exp\Bigr(-N\ABS{M_{jk}}^2\Bigr)
  =\exp\PAR{-N\mathrm{Tr}(MM^*)}.
  \]
  The spectral change of variables $M=U(D+N)U^*$, which is the Schur unitary
  decomposition, gives that the joint law of the eigenvalues of $M$ has
  density
  \begin{equation}\label{eq:phinn}
    \vphi^{N,N}(z_1,\ldots,z_n)
    :=\frac{N^{\frac{N(N+1)}{2}}}{1!2!\cdots N!}
    \frac{\e^{-\sum_{i=1}^N N\ABS{z_i}^2}}{\pi^N}
    \prod_{i<j}\ABS{z_i-z_j}^2
  \end{equation}
  with respect to the Lebesgue measure on $\dC^N$. This law is usually
  referred to as the ``complex Ginibre Ensemble'', see
  \cite{MR0173726,MR2641363,MR2932638,MR2908617}. This matches $P^N$ with
  \eqref{eq:Ginibre} with $\be_N=N^2$ so that the density of $P^N$ on
  $(\dR^2)^N=\dC^N$ can be written as
  \begin{equation}\label{eq:pnphinn}
    \frac{\dd P^N(z_1,\ldots,z_N)}{\dd z_1\cdots \dd z_N}
    =\vphi^{N,N}(z_1,\ldots,z_N).
  \end{equation}
  It is a well known fact -- see \cite[p.~271]{MR2129906},
  \cite[p.~150]{MR841088}, or \cite{MR2641363,MR2552864} -- that for every
  $1\leq k\leq N$, the $k$-th dimensional marginal distribution $P^{k,N}$ of
  $P^N$ has density
  \begin{align}
    \vphi^{k,N}(z_1,\ldots,z_k)
    &=\int_{\dC^{N-k}}\!\vphi^{N,N}(z_1,\ldots,z_N)\,\dd z_{k+1}\cdots \dd z_N
    \nonumber\\
    &=\frac{(N-k)!}{N!}\frac{\e^{-N(|z_1|^2+\cdots+|z_k|^2)}}{\pi^kN^{-k}}
    \det\SBRA{(\e_N(Nz_i\OL{z}_j))_{1\leq i,j\leq k}},\label{eq:phikn}
  \end{align}
  where $\e_N(w):=\sum_{\ell=0}^{N-1}w^\ell/\ell!$ is the truncated
  exponential series. The energy $H$ is a quadratic functional of the
  empirical measure $\mu^N:=\frac{1}{N}\sum_{i=1}^N\de_{x_i}$ of the
  particles:
  \begin{align}\label{eq:EmuN}
    H(x_1,\ldots,x_N)
    &=\int\!V(x)\,\mu^N(\dd x)
    +\frac{1}{2}\iint_{\neq}\!W(x-y)\,\mu^N(\dd x)\mu^N(\dd y)\\
    &=:\cE_{\neq}(\mu^N)\nonumber
  \end{align}
  where ``$\neq$'' indicates integration outside the diagonal. Since $\be_N\gg
  N$ as $N\to\infty$, under ${(P^N)}_N$, the sequence of empirical measures
  ${(\mu^N)}_N$ satisfies a large deviation principle with speed ${(\be_N)}_N$
  and good rate function $\cE-\inf\cE$ where $\cE$ is given for nice
  probability measures $\mu$ on $\dR^2$ by
  \[
  \cE(\mu):=\int\!V(x)\,\mu(\dd x) %
  +\frac{1}{2}\iint\!W(x-y)\,\mu(\dd x)\mu(\dd y).
  \]
  See for instance \cite{MR1606719,MR1746976,MR3262506} and references
  therein. The functional $\cE$ is strictly convex where it is finite, lower
  semi-continuous with compact level sets, and it achieves its global minimum
  for a unique probability measure $\mu_\infty$ on $\dR^2$, which is the
  uniform distribution on the unit disc with density
  $z\in\dC\mapsto\pi^{-1}\mathbf{1}_{\ABS{z}\leq1}$. From the large deviation
  principle it follows that almost surely
  \begin{equation}\label{eq:cirlaw}
    \mu^N\underset{N\to\infty}{\longrightarrow}\mu_\infty:=\arg\inf\cE
  \end{equation}
  weakly, regardless of the way we put ${(P^N)}_N$ in the same probability
  space.

\item \textbf{Crossover regime: $\be_N=N$.} In this case $P^N$ has density
  proportional to
  \[
  (x_1,\ldots,x_N)\in D%
  \mapsto \e^{-\sum_{i=1}^N\ABS{x_i}^2}\prod_{1 \leq i<j \leq N}\ABS{x_i-x_j}^{\frac{2}{N}}.
  \]
  We do not have a determinantal formula as in \eqref{eq:phikn}, and this gas
  is not associated with a standard random matrix ensemble. It is a two
  dimensional analogue of the one dimensional gas studied in
  \cite{2012PhRvL.109i4102A} leading to a Gauss-Wigner crossover. Following
  \cite{MR3262506}, we can expect that under ${(P^N)}_N$ the sequence of
  empirical measures ${(\mu^N)}_N$ satisfies a large deviation principle with
  speed ${(\be_N)}_N$ and rate function $\Wt\cE-\inf\Wt\cE$; here $\Wt\cE$ is
  given for every probability measure $\mu$ on $\dR^2$ by
  \[
  \Wt\cE(\mu):=\cE(\mu)-\cS(\mu)
  \quad\text{where}\quad
  \cS(\mu):=-\int\!\frac{\dd\mu}{\dd x}\log\frac{\dd\mu}{\dd x}\,\dd x
  \]
  when $\mu$ is absolutely continuous with respect to the Lebesgue measure,
  while $\cS(\mu):=+\infty$ otherwise. $\cS$ is the so-called
  Boltzmann\,--\,Shannon entropy. The minimizer of $\Wt\cE$ is no longer
  compactly supported but can still be characterized by Euler\,--\,Lagrange
  equations, and is a crossover between the uniform law on the disc and the
  standard Gaussian law on $\dR^2$. See \cite{MR3262506} for the link with
  Sanov's large deviation principle.
\end{itemize}

\subsubsection{Dyson Brownian Motion}

If we start with an $N\times N$ random matrix 
\[
M_t={(M_t^{j,k})}_{1\leq j,k\leq N}
\]
with i.i.d.\ entries following the diffusion
$\dd M_t^{j,k}=\dd B_t^{j,k}-N\na V(M_t^{j,k})\,\dd t$ then the eigenvalues in
$\dC=\dR^2$ of $M_t$ will not match our diffusion $X$ solution of
\eqref{eq:SDE}. This is due to the fact that $M_t$ is not a normal matrix in
the sense that $M_tM_t^*\neq M_t^*M_t$ with probability one as soon as $M_t$
has a density. In fact the Schur unitary decomposition of $M_t$ writes
$M_t=U_tT_tU_t^*$ where $U_t$ is unitary and $T_t=D_t+N_t$ is upper
triangular, $D_t$ is diagonal, and $N_t$ is nilpotent. The dynamics of $D_t$
is perturbed by $N_t$. The dynamics \eqref{eq:SDE} is not the analogue of the
Dyson Brownian motion, the process of the eigenvalues associated with the
Gaussian Unitary Ensemble, the one-dimensional log-gas studied in
\cite{MR2760897,xiangdongli}. We refer to \cite{MR3514219,bourgade-dubach} and
references therein for more information on this topic.

\subsubsection{Initial conditions}

In the case of the one-dimensional log-gas known as the Dyson Brownian Motion,
the stochastic differential equation still admits a unique strong solution
when the particles coincide initially. This is proved in
\cite[Prop.~4.3.5]{MR2760897} by crucially using the ordered particle system.
Unfortunately, it does not seem possible to extend such an argument to higher
dimensions. But it is likely that at least weak well-posedness should still
hold for our model.

\subsubsection{Arbitrary dimension, confinement, and interaction}

As in \cite{MR3262506}, many aspects should remain valid in arbitrary
dimension $d\geq2$, with a Coulomb repulsion and a more general
confinement~$V$. For instance, by analogy with the case without interaction
studied in \cite[Th.~2.2.19]{MR2352327}, it is natural to expect that Theorem
\ref{th:boum} remains valid beyond the quadratic confinement case, for example
in the quadratic ``dispersive'' case $V(x)=-|x|^2$, and in confined cases for
which $V(x)\to+\infty$ as $x\to\infty$ with polynomial growth. Nevertheless,
our choice is to entirely devote the present article to the two-dimensional
quadratic confinement case: this model is probably the richest in structure,
notably due to its link with the Ginibre Coulomb gas, which is a remarkable
exactly solvable model.

The model with non-singular interaction has extensively been studied in
arbitrary dimension, in relation with McKean\,--\,Vlasov equations, see
\cite{MR1847094,MR1431299,MR1108185} and references therein. The model in
dimension $d=1$ with logarithmic singular interaction has also extensively
been studied, see for instance
\cite{MR1440140,MR1698948,MR1872742,MR2062570,xiangdongli} and references
therein. See also \cite{berman-onnheim}.

\subsubsection{Logarithmic Sobolev inequality and other functional inequalities}

It is natural to ask whether $P^N$ satisfies a logarithmic Sobolev inequality,
which is stronger than the Poincaré inequality with half the same constant,
see \cite{MR1845806,MR3155209}. Indeed, for $P^N$, a Lyapunov approach is
probably usable by following the lines of \cite[Proof of
Prop.~3.5]{cattiaux-guillin}, see also \cite{MR2498560}, but there are
technical problems due to the shape of $D$ which comes from the singularity of
the interaction. Observe that the one-particle marginal $P^{1,N}$ satisfies
indeed a logarithmic Sobolev inequality with a constant uniform in $N$, as
mentioned in Remark \ref{rk:P1NLSI} after the proof of Theorem
\ref{th:Bobkov}.

Still about functional inequalities, the study of concentration of measure for
Coulomb gases in relation with Coulomb transport inequalities is considered in
the recent work \cite{chafai-hardy-maida}.

\subsubsection{Mean-field limit}\label{Mean-field limit}

In the regime $\beta_N = N^2$, by~\eqref{eq:cirlaw}
the empirical measure $\mu^N$ under $P^N$ tends to $\mu_{\infty}$ as
$N\to\infty$. 
More generally, when
the law of $X_0$ is exchangeable and for general $\beta_N$, one can ask about the behavior of the
empirical measure of the particles
$\mu^N_t:=\frac{1}{N}\sum_{i=1}^N\de_{X_t^i}$ as $N\to\infty$ and as
$t\to\infty$. This corresponds to study the following scheme:
\[
P^N_t \underset{t\to\infty}{\longrightarrow} P^N
\quad\text{and}\quad
\begin{array}{ccc}
  \mu^N_t &  \underset{t\to\infty}{\longrightarrow} & \mu^N \\
  \downarrow& & \downarrow\\
  \mu_{t} & \underset{t\to\infty}{\longrightarrow} & \mu_{\infty}
\end{array}
\]
for a suitable deterministic limit $\mu_{\infty}$.

At fixed $N$, the limit $\lim_{t\to\infty}P^N_t=P^N$, valid for an arbitrary
initial condition $X_0=x$, corresponds to the ergodicity phenomenon for the
Markov process $X$, quantified by the Poincaré inequality of Theorem
\ref{th:Lyapunov}. By the mean-field structure of \eqref{processX} and
\eqref{eq:gen}, it is natural to expect that if
\[
\si:=\lim_{N\to\infty}\frac{\al_N}{\be_N}\in[0,+ \infty)
\]
then the sequence ${({(\mu^N_t)}_{t\geq0})}_{N}$ converges, as a continuous
process with values in the space of probability measures in $\dR^2$, to a
solution of the following McKean\,--\,Vlasov partial differential equation
with singular interaction:
\begin{equation}\label{eq:MV}
  \pd_t\mu_t =\si\De\mu_t+\na\cdot((\na V+\na W*\mu_t)\mu_t).
\end{equation}
The convergence of
${({(\mu_t^N)}_{t\geq0})}_N$ can be thought of as a sort of law of large
numbers.
This is well understood in the
one-dimensional case with logarithmic interaction, see for instance
\cite{MR1217451,MR1440140}, using tightness and characterization of the
limiting laws. However the uniqueness arguments used in one-dimension are no
longer valid for our model, and different ideas need to be developed, see
\cite{coulsim2}. We also refer to \cite{duerinckx} and references therein for the analysis of similar
evolution equations without noise and confinement. 

Theorem \ref{th:2ndmoment} suggests to take $\al_N=N$. Let us comment on the
couple of special cases already considered in our large deviation principle
analysis of ${(P^N)}_N$: $\be_N=N^2$ and $\be_N=N$, when $\al_N=N$.

\begin{itemize}
\item\textbf{Random matrix theory regime with vanishing noise: $\al_N=N$ and
    $\be_N=N^2$.} In this case $\si=0$ and the limiting McKean\,--\,Vlasov
  equation \eqref{eq:MV} does not have a diffusive
  part. 
  Since $\al_N=N$ we have a constant speed for the second moment evolution.
  Since $\be_N=N^2$ we have explicit determinantal formulas for $P^N$ from the
  complex Ginibre Ensemble \eqref{eq:phikn}. The absence of diffusion implies
  that if we start from an initial state $\mu_0$ which is supported in a line,
  then $\mu_t$ will still be supported in this line for any $t\in[0,\infty)$,
  and will thus never converge as $t\to\infty$ to the uniform distribution on
  the unit disc of the complex plane. In particular, the long time equilibrium
  depends clearly on the initial condition.
\item\textbf{Crossover regime with non-vanishing noise: $\al_N=N$ and
    $\be_N=N$.} In this case $\si=1$ and the McKean\,--\,Vlasov equation
  \eqref{eq:MV} has a diffusive term. This regime is also considered in
  \cite{MR1145596,MR1362165} for instance, see also \cite{MR3254330}. The
  Keller\,--\,Segel model studied in
  \cite{fournier-jourdain,cattiaux-pedeches} is the analogue with an
  attractive interaction instead of repulsive.
\end{itemize}

\section{Useful formulas}
\label{se:formulas}

In this section we gather several useful formulas related to the energy $H =
H_V + H_W$ and the operator $L,$ defined in~\eqref{eq:E} and~\eqref{eq:gen}
respectively. Recall that $V(z) = \vert z \vert^2$ and
$W(z)=-2\log\ABS{z}$ on $\dR^2 \setminus\{0\}$, giving 
\[
\na V(z)=2z
\quad\text{and}\quad
\na W(z)=-\frac{2z}{\ABS{z}^{2}}.
\]
Moreover we let $\vert x \vert^2 = \sum_{i=1}^N \vert x_i \vert^2$ for $x =
(x_1, \dots, x_N) \in (\dR^2)^N$.

\medskip
\noindent
{\bf Gradient.} 
By~\eqref{eq:E}, for any $x\in D$ and $i\in\{1,\ldots,N\}$,
\begin{equation}\label{eq:naHV}
\na_{x_i} H_V(x) = \frac{1}{N}\na V(x_i) = \frac{2}{N} x_i
\end{equation}
and
\begin{equation}\label{eq:naHW}
\na_{x_i} H_W(x)
=
\frac{1}{N^2}\sum_{j\neq i}\na W(x_i-x_j)
=
-\frac{2}{N^2}\sum_{j\neq i}\frac{x_i-x_j}{\ABS{x_i-x_j}^2}.
\end{equation}

\medskip
\noindent {\bf Hessian.} By~\eqref{eq:naHV}-\eqref{eq:naHW}, for any $x\in D$
and $i,j\in\{1,\ldots,N\}$,
\[
\na^2_{x_i,x_j} H_V(x) %
=\begin{cases}
  \frac{1}{N}\na^2 V(x_i) %
  &\text{if $i=j$}\\
0
  &\text{if $i\neq j$}
\end{cases}\\
\]
and
\[
\na^2_{x_i,x_j} H_W(x) %
=\begin{cases}
  +\frac{1}{N^2}\sum_{k\neq i}\na^2 W(x_i-x_k)
  &\text{if $i=j$}\\
  -\frac{1}{N^2}\na^2 W(x_i-x_j)
  &\text{if $i\neq j$}.
\end{cases}\\
\]
This gives
\begin{equation}\label{eq:hessE}
  \na^2 H_V=\frac{2}{N}I_{2N} %
  \qquad \textrm{and} \qquad %
  \na^2 H_W = \frac{1}{N^2}A
\end{equation}
where $I_{2N}$ is the $2N\times 2N$ identity matrix and $A$ is a $N\times N$
bloc matrix with diagonal and off-diagonal $2\times 2$ blocs
\[
A_{i,i}=\sum_{k\neq i}\na^2W(x_i-x_k),
\quad
A_{i,j}=-\na^2W(x_i-x_j), \quad i\neq j.
\]

\medskip
\noindent {\bf Operator.} The generator $L$ defined in~\eqref{eq:gen} on
functions $f: D \to \dR$ is given by
 \begin{eqnarray}\label{eq:Lbis}
  L f(x)
  &=&
  \frac{\al_N}{\be_N}\De f(x)
  -\frac{\al_N}{N}\sum_{i=1}^N\na V(x_i)\cdot\na_{x_i}f(x)
  -\frac{\al_N}{N^2}\sum_{1 \leq i\neq j \leq N}\na W(x_i-x_j)\cdot\na_{x_i} f(x) %
  \nonumber\\
 &=&
\frac{\al_N}{\be_N}\De f(x)
-2\frac{\al_N}{N}\sum_{i=1}^Nx_i\cdot\na_{x_i}f(x)
+2\frac{\al_N}{N^2}\sum_{1 \leq i\neq j \leq N}
\frac{(x_i-x_j)\cdot\na_{x_i} f(x)}{|x_i-x_j|^2}.  
\end{eqnarray}

Let us compute now $LH_V$, $LH_W$, and $LH$. First of all, since $\nabla W$ is
odd, we get by symmetrization from \eqref{eq:naHV}-\eqref{eq:naHW} that
\begin{align}
  \ABS{\na H}^2(x) %
  &=\frac{1}{N^2}\sum_{i=1}^N\bigr|\na V(x_i)\bigr|^2
  +\frac{1}{N^4}\sum_{i=1}^N\Bigr|\sum_{j\neq i}\na W(x_i-x_j)\Bigr|^2
  \nonumber\\
  & \qquad +\frac{1}{N^3}\sum_{i\neq j} %
  \PAR{\na V(x_i)-\na V(x_j)} \cdot \na W(x_i-x_j)\nonumber\\
  &=\frac{4}{N^2}\ABS{x}^2
  +\frac{4}{N^4}\sum_{i=1}^N\Bigr|%
  \sum_{j\neq i}\frac{x_i-x_j}{\ABS{x_i-x_j}^2}\Bigr|^2
  -4\frac{N-1}{N^2}.
 \label{eq:naE2}
\end{align}
Moreover, from \eqref{eq:hessE} and $\De W=0$ on $D$, we get
\begin{equation}\label{eq:deE}
  \De H_W(x)=0
  \quad\text{and}\quad
  \De H(x)
  =\De H_V(x)
  =\sum_{i=1}^N\frac{1}{N}\De V(x_i)
  =4.
\end{equation}
By~\eqref{eq:naHV}-\eqref{eq:naHW} and by symmetry we also have
\begin{align}\label{eq:LHV}
  L {H_V}(x)
  &=\frac{\al_N}{\be_N}\De {H_V}(x)-\al_N\na H(x)\cdot\na {H_V}(x)\nonumber\\
  &=4\frac{\al_N}{\be_N}-\frac{4\al_N}{N^2}\ABS{x}^2%
  +\frac{4\al_N}{N^3}%
  \sum_{i=1}^N\sum_{j\neq i}\frac{x_i-x_j}{\ABS{x_i-x_j}^2}\cdot x_i\nonumber\\
  &=4\frac{\al_N}{\be_N}-\frac{4\al_N}{N^2}\ABS{x}^2
  +\frac{2\al_N}{N^3}%
  \sum_{1 \leq i\neq j \leq N}\frac{x_i-x_j}{\ABS{x_i-x_j}^2}\cdot(x_i-x_j)\nonumber\\
  &=4\frac{\al_N}{\be_N}+ 2 \al_N\frac{N-1}{N^2} -\frac{4\al_N}{N}{H_V}(x)
\end{align}  
and likewise
\begin{align}
  L{H_W}(x)
  &=\frac{\al_N}{\be_N}\De H_W(x)-\al_N\nabla H(x)\cdot\nabla H_W(x)\nonumber\\
  &=2\al_N \frac{N-1}{N^2} -4\frac{\al_N}{N^4} \sum_{i=1}^N\bigg|\sum_{j\neq
    i}\frac{x_i-x_j}{\ABS{x_i-x_j}^2}\bigg|^2\label{eq:LHW}
\end{align}
From \eqref{eq:LHV} and \eqref{eq:LHW} we finally get
\begin{align}
  LH(x)
  &= \, L{H_V}(x)+L{H_W}(x)\nonumber\\
  &=
  4\frac{\al_N}{\be_N} %
  +4\al_N\bigg(\frac{N-1}{N^2}-\frac{1}{N}{H_V}(x)-\frac{1}{N^2}\sum_{i=1}^N%
  \bigg|\frac{1}{N}\sum_{j\neq
    i}\frac{x_i-x_j}{\ABS{x_i-x_j}^2}\bigg|^2\bigg).\label{eq:LH}
\end{align}

Note that the fact that the singular repulsion potential $W$ is the
fundamental solution of the diffusion part $\De$ simplifies the expression of
$LH$, in contrast with the situation in dimension $1$ studied in
\cite[p.~559]{MR1217451}, see also \cite{xiangdongli}.

\section{Proof of Theorem \ref{th:boum}}
\label{se:boum}

\begin{lemma}[Connectivity]\label{le:connectivity}
  The set $D$ defined by \eqref{eq:D} is path-connected in $(\dR^2)^N$.
\end{lemma}


\begin{proof}[Proof of Lemma \ref{le:connectivity}]
  It suffices to show that for any $x:=(x_1,\ldots,x_N)\in D$ and
  $y:=(y_1,\ldots,y_N)\in D$, there exists a continuous map
  $\gamma:=(\gamma_1,\ldots,\gamma_N):[0,1]\mapsto D$ such that $\gamma(0)=x$
  and $\gamma(1)=y,$ which must be understood as the position in time of $N$
  moving particles in space. This corresponds to move a cloud of $N$ distinct
  and distinguishable particles into another cloud of $N$ distinct and
  distinguishable particles. Let us proceed by induction on $N$. The property
  is immediate for $N=1$. Suppose that $N\geq1$ and assume that one has
  already constructed $t\in[0,1]\mapsto(\gamma_1(t),\ldots,\gamma_N(t))$. One
  can first construct $\gamma_{N+1}$ in such a way that
  $\{t\in[0,1]:\gamma_{N+1}(t)\in\{\gamma_1(t),\ldots,\gamma_N(t)\}\}$ is a
  finite set. Second, one may modify the path $\gamma_{N+1}$, locally at the
  intersection times by varying the speed, in order to make this set empty.
  This is possible since $d=2$, and possibly impossible if $d=1$ since a
  particle cannot bypass another one.
\end{proof}

\begin{lemma}[Coercivity]\label{le:coercivity}
  For any fixed $N$, we have $H\geq 0$, 
  \[
  \lim_{x\to\partial D}H(x)=+\infty,
  \]
  and $\e^{-\be H}$ is Lebesgue integrable on $D$ for any $\be>0$.
\end{lemma}

\begin{proof}[Proof of Lemma \ref{le:coercivity}]
  Let $x = (x_1, \dots, x_N)$ in $D$. Then
  \[
  \frac{1}{2} \sum_{i \neq j} \vert x_i - x_j \vert^2 
  =  
  \frac{1}{2} \sum_{i, j=1}^N \vert x_i - x_j \vert^2 
  =
  N \sum_{i=1}^N \vert x_i \vert^2 - \Bigr\vert \sum_{i=1}^N x_i \Bigr\vert^2 
  \leq N|x|^2
  \]
  so for $u_{ij} = \vert x_i - x_j \vert^2 $ it holds
  \[
  2N^2 \, H(x) 
  = N \vert x \vert^2 
  + N \vert x \vert^2 
  - \sum_{i \neq j} \log u_{ij} \geq N \vert x \vert^2 
  + \sum_{i \neq j} \Big( \frac{u_{ij}}{2} 
  - \log u_{ij} \Big).
  \]
  But $u/2 - \log u \geq 1 - \log(2) \geq 1/4$ for all $u>0$, so
  \begin{equation}\label{minoH} 
    H(x) %
    \geq \frac{\vert x \vert^2}{2N} + \frac{1}{2N^2}\frac{N(N-1)}{4} %
    \geq \frac{\vert x \vert^2}{2N} + \frac{1}{16} \cdot
  \end{equation}
  In particular $H \geq 0$ and $\e^{-\be H}$ is Lebesgue integrable on $D$ for
  any $\be>0$.

  We now prove that $H(x)\to+\infty$ as $x\to\partial D$. It suffices to show
  that for any $R>0$ there exists $A>0$ and $\veps>0$ such that $H(x)\geq R$
  as soon as $\max_{1\leq i\leq N}\ABS{x_i}\geq A$ or $\min_{1\leq i\neq j\leq
    N}\ABS{x_i-x_j}\leq\veps$. First, let us fix $R>0$. Then,
  by~\eqref{minoH}, $H(x) \geq R$ as soon as $\vert x \vert^2 \geq 2NR$,
  giving such an $A$.
  
  Then, for $\veps>0$ to be chosen later, assume that for some $i\neq j$ we
  have $\ABS{x_i-x_j} \leq\veps.$ Then, by definition of $H(x)$,
  \[
  N^2H(x)\geq
  2\log\frac{1}{\ABS{x_i-x_j}}
  +\sum_{\SSK{1 \leq k\neq l \leq N\\\{k,l\}\neq\{i,j\}}}\log\frac{1}{\ABS{x_k-x_l}}.
  \]
  We can assume that $\max_{1\leq r\leq N}\ABS{x_r}\leq A$ otherwise we have
  already seen that $H(x)\geq R$. Hence, for any $(k,l)$ with $k\neq l$
  we have 
  \[
  \log\frac{1}{\ABS{x_k-x_l}}
  \geq-\log(1+\ABS{x_k})-\log(1+\ABS{x_l}) %
  \geq -2\log(1+A)
  \]
  using the inequality $|a-b|\leq(1+|a|)(1+|b|)$ for $a,b\in\dC.$ As a
  consequence
  \[
  N^2H(x)\geq -2\log\veps-2N^2\log(1+A),
  \]
  which is $\geq R$ for a small enough $\veps$.
\end{proof}

In the sequel we use the notation $\dP_x=\dP[\,\cdot\mid X_0=x]$ and $\dE_x =
\dE[\,\cdot \mid X_0=x]$.

\begin{proof}[Proof of Theorem \ref{th:boum}]
  We first construct the process $X$ starting in $D$ up to its explosion time.
  Given an initial condition $x\in D$, for each $\varepsilon\in (0,
  \min_{1\leq i\leq N}|x^i-x^j| )$ we consider a smooth function
  $W^{\varepsilon}$ on $\dR^2$ coinciding with $W$ on $\{z\in \dR^2: \,
  |z|\geq \varepsilon\}$ and we set
  \[
  H^{\varepsilon}=H_V+H_{W^{\varepsilon}}.
  \]  
  Given a Brownian motion $B$ in a fixed probability space we let
  $X^{\varepsilon}$ denote the unique pathwise solution to the stochastic
  differential equation
  \begin{equation}\label{eq:SDE_vareps}
    \dd X^{\varepsilon}_t
    =\sqrt{2\frac{\al_N}{\be_N}}\dd B_t
    -\al_N\na H^{\varepsilon}(X^{\varepsilon}_t)\,\dd t, \qquad X^{\varepsilon}_0 = x.
  \end{equation}
  Notice that for $\varepsilon'\in (0,\varepsilon]$, the processes
  $X^{\varepsilon}$ and $X^{\varepsilon'}$ coincide up to the stopping time
  \[
    T^{\varepsilon,\varepsilon'}
    =\inf\{s\geq 0: %
    \min_{i\neq j}  |(X^{\varepsilon'}_s)^i - (X^{\varepsilon'}_s)^j| %
    \leq \varepsilon\}.
  \]
  For each $\varepsilon \in (0, \min_{1\leq i\leq N}|x^i-x^j| )$ we can thus
  unambiguously define a stopping time $T^{\varepsilon}$ and a process $X$ on
  $[0,T^{\varepsilon}]$, setting
  $T^{\varepsilon}=T^{\varepsilon,\varepsilon'}$ and $X=X^{\varepsilon'} $ for
  any $\varepsilon'\in (0,\varepsilon)$. By continuity, we have
  $T^{\varepsilon'}>T^{\varepsilon}$ a.s., and so $X$ is uniquely defined up
  to the stopping time $ T_{\pd D}$ defined in \eqref{eq:Texp}. On the other
  hand, the process $X$ satisfies equation \eqref{processX} on each interval
  $[0, T^{\varepsilon} )$ and hence on $ [0, T_{\pd D})$ too. Thus, we just
  have to prove that $ T_{\pd D}=\infty$ a.s.

  Given $R>0$, define the stopping times
  \[
    T'_R:=\inf\{t\geq0:H(X_t)>R\}\in[0,\infty],
    \quad\text{and}\quad
    T':=\lim_{R\to\infty}T'_R=\sup_{R>0}T'_R\in[0,\infty].
  \]
  Lemma \ref{le:coercivity} gives
  $\{T'=\infty\}\subset\{T_{\partial D}=\infty\}$: indeed on $\{T'=\infty\}$,
  for every $t\geq0$ we have $\sup_{s\in[0,t]}H(X_s)<\infty$ ; by Lemma
  \ref{le:coercivity} this means that $T_{\partial D}=\infty$.
  
  Let us now show that $\dP_x (T' = \infty) = 1$. Thanks to \eqref{eq:LH}, we
  have $LH\leq c$ on $D$ for $c= 4 \alpha_N (1/\beta_N + 1/N).$ Moreover,
  given $R \geq 1$ and proceeding as in the end of the proof of Lemma
  \ref{le:coercivity} we can choose $\varepsilon <e^{- CR N^2 \log N}$ for a
  numerical constant $C$ such that the function $H$ (respectively $LH$)
  coincides with $H^{\varepsilon}$ (respectively $LH^{\varepsilon}$) along the
  trajectory of $X$ on the interval $[0, T'_R]$; we can therefore apply the
  Itô formula to $X_{t\wedge T'_R}$ and $H^{\varepsilon}$ to get that
  \begin{equation}\label{eq:ito}
    \dE_x(H(X_{t\wedge T'_R}))-H(x)
    =\dE_x\PAR{\int_0^{t\wedge T'_R}\!LH(X_s)\,\dd s}
    \leq
    \dE_x\PAR{\int_0^{t\wedge T'_R}\! c\,\dd s} 
    \leq ct, 
  \end{equation}
  for each $t\geq 0$.  In particular
  \[
    \sup_{R >0} \dE_x(H(X_{t\wedge T'_R}))< \infty.
  \]
  On the other hand, since $H$ is everywhere nonnegative by Lemma
  \ref{le:coercivity}, we have
  \[
    R \, \IND_{T'_R \leq t}   \leq H(X_{t\wedge T'_R}),
  \]
  from which it follows that
  \[
    \dP_x (T'_R \leq t) \leq \frac{1}{R} \sup_{R >0} \dE_x  (H(X_{t\wedge T'_R})).
  \]
  Finally
  $\dP_x (T' \leq t) = \lim_{R \to \infty} \dP_x (T'_R \leq t) =0,$ for
  any $t\geq0$, and thus $\dP_x (T'=\infty) = 1.$
 \end{proof}

 Note that our proof of non-explosion notably differs from the one of
 \cite{MR1217451} and \cite{xiangdongli}: we deal with $\partial D$ at once,
 instead of handling separately $\infty$ and $x_i\neq x_j$, thanks to the
 geometric Lemma \ref{le:coercivity}.

\begin{remark}\label{XvarepstoX}\ %
  \begin{itemize}
  \item[a)] From the previous proof we see that the process $X$ and the process
    $X^{\varepsilon}$ as in \eqref{eq:SDE_vareps} coincide up to the stopping
    time $T^{\varepsilon}$. Moreover, $T^{\varepsilon}\to \infty$ a.s.\ as
    $\varepsilon\to 0$. This readily implies that $X^{\varepsilon}\to X$ a.s.
    uniformly on each finite time interval $[0,T]$ and, in particular that
    $\mathrm{Law}(X^{\varepsilon})\to \mathrm{Law}(X)$ in
    $C([0,T],(\dR^2)^N)$.
  \item[b)] Since $H$ is bounded from below and $LH$ is bounded from above,
    letting $R\to \infty$ in the first equality in \eqref{eq:ito} and using
    twice Fatou's Lemma we get that
    \[ 
    \dE_x(H(X_t))-H(x)
    \leq \dE_x\PAR{\int_0^{t}\!LH(X_s)\,\dd s}
    \]
    with  both sides finite, for all $t\geq 0$.
  \end{itemize} 
\end{remark}

  \section{Proof of Theorem \ref{th:2ndmoment}}\label{se:2ndmoment}

\begin{proof}[Proof of Theorem \ref{th:2ndmoment}]
  By the It\^o formula and~\eqref{eq:LHV}, $N^{-1}\ABS{X_t}^2 = {H_V}(X_t)$
  evolves according to the stochastic differential equation
  \begin{equation}\label{eq:ItoH_V}
    \begin{split}
      \dd{H_V}(X_t)
      &=L{H_V}(X_t)\dd t+\sqrt{2\frac{\al_N}{\be_N}}\nabla {H_V}(X_t)\,\dd B_t \\
      &=\PAR{4\frac{\al_N}{\be_N} %
        +2\al_N\frac{N-1}{N^2}-\frac{4\al_N}{N}{H_V}(X_t)}\dd t
      +\sqrt{2\frac{\al_N}{\be_N}}\frac{2}{N} X_t\,\dd B_t.\\
    \end{split}
  \end{equation}
  The process $ {H_V}(X_t)$ thus satisfies, until the first time it hits $0$,
  the stochastic differential equation \eqref{SDECIR} with the Brownian motion
  $b_t$ defined by $\dd b_t= \frac{X_t \cdot \, \dd B_t}{|X_t|} $. Standard
  properties of the CIR process (see \cite{MR785475}) and the fact that
  $4\frac{\al_N}{\be_N} +2\al_N\frac{N-1}{N^2} \geq 4 \frac{\al_N}{N \be_N}$,
  imply this stopping time is $\infty$ a.s. Pathwise uniqueness for
  \eqref{SDECIR} ensures that the law of $ {H_V}(X_t)$ is the same as for the
  CIR process (in particular, its invariant distribution is given in
  \cite{MR785475}).
  
  Ergodicity of the solution $R$ to \eqref{SDECIR} is proved in
  \cite{MR3167406} by a non quantitative approach. Let us prove the long time
  convergence bound~\eqref{eq:cvCIR} in Wasserstein$-1$ distance. By standard
  arguments, it is enough to show that for any pair $(R^x_t,R^y_t)$ of
  solutions to \eqref{SDECIR} driven by the same (fixed) Brownian motion
  $b_t$, and such that $(R^x_0,R^y_0)=(x, y)$, one has
  \[
  \dE[|R^x_t-R^y_t|]\leq\e^{-4\frac{\al_N}{N}  t}  |x-y|.
  \]
  This can be done adapting classical uniqueness argument for square root
  diffusions found in \cite{MR637061}. Indeed, consider the function
  \[
  x\in\dR_+\mapsto\rho(x):= \sqrt{\frac{8 \alpha_N}{ N \beta_N} x} 
  \]
  and the sequence $\{a_{\ell}\}_{\ell\geq 1}$ defined as 
  \[
  a_{0} = 1\quad\text{and}\quad 
  a_{\ell}= a_{\ell-1}\e^{-\ell{\frac{8 \alpha_N}{N \beta_N}}},\ \ell\geq
  1.
  \]
  Note that $a_{\ell}\searrow 0$ and $\int_{a_{\ell}}^{a_{\ell-1}} \rho(z)^{-2
  }\dd z = \ell$. For each $\ell \geq 1$, let moreover $z \mapsto
  \psi_{\ell}(z)$ be a non-negative continuous function supported on
  $(a_{\ell} , a_{\ell-1})$ such that $\int_{a_{\ell}}^{a_{\ell-1}}
  \psi_{\ell}(z)\dd z = 1$ and $0 \leq \psi_{\ell}(z)\leq
  2\ell^{-1}\rho(z)^{-2}$ for $a_{\ell} < z < a_{\ell-1}$. Consider also the
  even non-negative and twice continuously differentiable function $\phi_\ell$
  defined by
  \[
  \phi_{\ell}(x) 
  =\int_{0}^{|x|}d y\int_{0}^{y}\psi_{\ell}(z)\dd z, \quad x \in \dR
  \]
  For all $x\in \dR$ it satisfies : $\phi_{\ell}(x) \nearrow |x| , \,
  \phi'_{\ell}(x) \to \mbox{sign}(x) $ as $\ell \rightarrow \infty$, $ 0 \leq
  \phi'_{\ell}(x)x \leq |x| $ and $0\leq \phi''_{\ell}(x) \frac{8 \alpha_N}{ N
    \beta_N} |x| \leq 2\ell^{-1} $. Applying the It\^o formula to $\phi_\ell$
  and $\zeta_t:= R^x_t-R^y_t $ we get
  \begin{equation*}
    \begin{split}
      \phi_\ell(\zeta_{t})= 
      & \,   M^{\ell}_{t} 
      -  4\frac{\al_N}{N}\int_0^{t}\phi'_\ell(\zeta_s) \zeta_s\,\dd s  
      +  4\frac{\al_N}{N \be_N}\int_0^{t}\phi''_\ell(\zeta_s) \zeta_s\,\dd s\\
    \end{split}
  \end{equation*}
  for some martingale $M_{t}^{\ell}$. Taking expectation, letting $\ell \to
  \infty$ and applying Gronwall's lemma, the desired inequality is obtained.
  Assertion~\eqref{eq:mom2evol} follows from \eqref{eq:ItoH_V}, noting that
  the function $f(t) = \dE[H_V(X_t)\mid X_0=x]$ solves
  \[
  f(t) = f(0) 
  + \int_0^t \Big( 4\frac{\al_N}{\be_N} 
  +2\al_N\frac{N-1}{N^2} 
  - \frac{4\al_N}{N} f(s) \Big)\,\dd s
  \]
  for all $t>0$, and integrating this equation.
\end{proof}

\section{Proof of Theorem \ref{th:Lyapunov}}
\label{se:Lyapunov}

\begin{proposition}[Lack of convexity]\label{pr:noconvex}
  The set $D$ defined by \eqref{eq:D} is not convex. Moreover, the Hessian
  matrix of the function $H$ is not always positive definite on $D$.
\end{proposition}

\begin{proof}[Proof of Proposition \ref{pr:noconvex}]
  The set $D$ is not convex since $0\in[-x,x]\cap D^c$ for any $x\in D$.

  The convexity of $H$ could be studied using a bloc version of the
  Ghershgorin theorem, see \cite{MR0151473}, if $W$ were convex. Unfortunately
  it turns out that $W$ is nowhere convex. More precisely, setting
  $z=(a,b)^\top\in\dR^2\setminus\{(0,0)\}$, we get
  \[
  W(z)
  =-\log(\ABS{z}^2)
  =-\log(a^2+b^2)
  \]
  and
  \[
  \na W(z)
  =-2\frac{z}{\ABS{z}^2}
  =-2\frac{(a,b)^\top}{a^2+b^2}
  \quad
  \text{and}
  \quad
  \na^2W(z)
  =
  2\frac{
    \begin{pmatrix}
      a^2-b^2 & 2ab \\
      2ab & b^2-a^2
    \end{pmatrix}
  }{(a^2+b^2)^2}.
  \]
  Thus
  \[
  \mathrm{Tr}(\na^2 W(z))=0
  \quad\text{and}\quad
  \det(\na^2 W(z))
  =-\frac{4}{\ABS{z}^4}.
  \]
  Consequently the two eigenvalues $\la_\pm(z)$ of $\na^2W(z)$ satisfy
  \[
  \la_-(z)=-\la_+(z)=-\frac{2}{a^2+b^2}=-\frac{2}{\ABS{z}^2}
  \underset{z\to0}{\longrightarrow}-\infty,
  \]
  and have respective eigenvectors $(-b, a)$ and $(a,b)$. In particular $W$ is
  not~convex.

\medskip

  Now, by \eqref{eq:hessE}, if we fix $x_1,\ldots,x_{N-1}$ and let
  $x_N$ tend to $x_1$, then $\nabla^2W(x_1-x_j)$ will remain bounded for any
  $j\in\{2,\ldots,N-1\}$ while the smallest eigenvalue of $\nabla^2W(x_1-x_N)$
  blows down to $-\infty$. Therefore $\nabla^2_{x_1,x_1}H(x_1,\ldots,x_N)$ and
  thus $\nabla^2H(x_1,\ldots,x_N)$ is not positive definite for such points.

  Note however that we may also use \eqref{eq:hessE} to get that
  $\nabla^2H(x_1,\ldots,x_N)$ is positive definite at points of $D$ for which
  all the differences $x_i-x_j$ are large enough.
\end{proof}

The following Lemma is the gradient version of Lemma \ref{le:coercivity}.

\begin{lemma}[Gradient coercivity]\label{le:coercivity-nabla}
For any $N$ and $x = (x_1, \dots, x_N)$ in $D$ we have
  \[
  \vert \nabla H (x) \vert^2 \geq 
  \frac{4}{N^2} \vert x \vert^2 %
  + \frac{4}{N^4} \sum_{i\neq j}\frac{1}{|x_i-x_j|^2 }- 4 \frac{N-1}{N^2} \cdot
  \] 
  In particular 
  \[
  \lim_{x\to\partial D}\ABS{\nabla H(x)}=+\infty.
  \]
\end{lemma}

\begin{proof}[Proof of Lemma \ref{le:coercivity-nabla}]
  This is a consequence of~\eqref{eq:naE2} and the fact that for any $N$ and
  any distinct $x_1,\ldots,x_N\in\dR^2$,
  \begin{equation}\label{eq:bolley-lemma}
    S_N:=\sum_{i=1}^{N}
    \Big\vert\sum_{j\neq i}\frac{x_i-x_j}{|x_i-x_j|^2}\Big\vert^2%
    -\sum_{i=1}^{N}\sum_{j\neq i}\frac{1}{|x_i-x_j|^2 }\geq0.
  \end{equation}

  For the proof of~\eqref{eq:bolley-lemma}, we first observe that
  \[
  S_2=0
  \]
  and we now consider $N\geq 3$ for which
  \[
  S_N=2\sum_{i=1}^{N}
  \sum_{\substack{1\leq j<k\leq N\\j,k\not=i}}
  \frac{(x_i-x_j) \cdot (x_i-x_k)}{\vert x_i-x_j\vert^2\vert x_i-x_k\vert^2}.
  \]
 
  Decomposing
  \[
  \sum_{i=1}^{N}
  \sum_{\substack{1\leq j<k\leq N\\j,k\not=i}}\cdot=\sum_{1\leq i<j<k\leq N}\cdot
  +\sum_{1\leq j<i<k \leq N}\cdot+\sum_{1 \leq j<k<i \leq N}\cdot
  \]
  and letting $I=j,J=i$ and $K=k$ in the second sum on the right-hand side and
  $I=j,J=k$ and $K=i$ in the third sum, we see that $S_N/2$ is equal to
  \begin{multline*}
    \sum_{1\leq i<j<k\leq N}
    \frac{\vert x_j-x_k\vert^2 (x_i-x_j)\cdot (x_i-x_k)}%
    {\vert x_i-x_j\vert^2\vert x_j-x_k\vert^2\vert x_k-x_i\vert^2}\\
    +\frac{\vert x_k-x_i\vert^2 (x_j-x_i)\cdot(x_j-x_k)}%
    {\vert x_i-x_j\vert^2\vert x_j-x_k\vert^2\vert x_k-x_i\vert^2}%
    +\frac{\vert x_i-x_j\vert^2 (x_k-x_i)\cdot(x_k-x_j)}%
    {\vert x_i-x_j\vert^2\vert x_j-x_k\vert^2\vert x_k-x_i\vert^2}\cdot
  \end{multline*}
  But
  \[
  \vert x_j-x_k\vert^2 %
  = \vert x_i-x_j\vert^2+\vert x_i-x_k\vert^2-2 \, ( x_i-x_j) \cdot (x_i-x_k)
  \]
  so
  \[
  S_N=4\sum_{1\leq i<j<k \leq N}
  \frac{\vert x_i-x_j\vert^2 \vert x_i-x_k\vert^2-(x_i-x_j)\cdot (x_i-x_k)^2}%
  {\vert x_i-x_j\vert^2\vert x_j-x_k\vert^2\vert x_k-x_i\vert^2} \cdot
  \]
  Hence $S_N\geq 0$ by the Schwarz inequality. This shows also that equality
  is achieved when $x_i-x_j$ and $x_i-x_k$ are parallel for any $i,j,k$ for
  instance when $x_i=(i,0)$ for any $i$, thanks to the equality case in the
  Schwarz inequality. Let us observe from the proof that the same bound would
  hold in any Hilbert space.
\end{proof}

The following lemma is the counterpart on $H_W$ of Theorem \ref{th:2ndmoment}
for $H_V$. It is likely that the bounds in the lemma are not optimal, as we
would expect bounds independent of $N$. This is probably due to our use of the
bound~\eqref{eq:bolley-lemma}. The lemma is not used but has its own interest
as we see that the particular speed $\alpha_N =N$ naturally appears in the
upper bounds, as in Theorem \ref{th:2ndmoment}.
 
\begin{lemma}[Energy evolution]\label{le:enevo}
  For every $x\in D$ and $t\geq0$, let us define
  \[
  \eta_x(t):=\frac{2N}{N-1}\dE_x [{H_W}(X_t)  ]
  \quad\text{where}\quad
  {H_W}(x):=\frac{1}{2N^2}\sum_{i\neq j}W(x_i-x_j).
  \]
  Then, for every $x\in D$ and $t\geq0$,
  \[
  \eta_x(t) %
  \leq-\log\Bigr(\e^{-\eta_x(0)-4 \alpha_Nt/N}+\frac{2}{N}(1-\e^{-4 \alpha_Nt/N})\Bigr)
  \]
  and in particular
  \[
  \eta_x(t)\leq\log\frac{N}{2(1-\e^{-4 \alpha_Nt/N})}
  \quad\text{and}\quad
\eta_x(t)\leq\max\bigg(\eta_x(0),\log\frac{N}{2}\bigg).
  \]
\end{lemma}

\begin{proof}[Proof of Lemma \ref{le:enevo}]
  Taking expectation to the first line in equation \eqref{eq:ItoH_V} and
  subtracting the obtained identity from the inequality in Remark
  \ref{XvarepstoX} b), we get
  \[ 
  \dE_x(H_W(X_t))-H_W(x)
  \leq \dE_x\PAR{\int_0^{t}\!LH_W(X_s)\,\dd s}
  \]
  for all $t\geq 0$. But from \eqref{eq:LHW} and \eqref{eq:bolley-lemma} we
  get
  \begin{align*}
  L{H_W}(x)
  &=2\frac{\al_N}{N^2}(N-1)
  -4\frac{\al_N}{N^4}
  \sum_{i=1}^N\bigg|\sum_{j\neq i}\frac{x_i-x_j}{\ABS{x_i-x_j}^2}\bigg|^2\\
  &\leq 2\frac{\al_N}{N^2}(N-1)
  -4\frac{\al_N}{N^4}\sum_{i\neq j}\frac{1}{\ABS{x_i-x_j}^2}\\
  &= 4\frac{\al_N}{N^2}\frac{N-1}{N}
  \SBRA{\frac{N}{2}-\frac{1}{N(N-1)}\sum_{i\neq j}\frac{1}{\ABS{x_i-x_j}^2}}.
  \end{align*}
  On the other hand, by the Jensen inequality,
  \[
  {H_W}(x)
  =\frac{N-1}{2N}\frac{1}{N(N-1)}\sum_{i\neq j}\log\frac{1}{\ABS{x_i-x_j}^2}
  \leq\frac{N-1}{2N}
  \log\bigg(\frac{1}{N(N-1)}\sum_{i\neq j}\frac{1}{\ABS{x_i-x_j}^2}\bigg).
  \]
  Therefore, we get
  \[
  L{H_W}(x)
  \leq 4\frac{\al_N}{N^2}\frac{N-1}{N}
  \SBRA{\frac{N}{2}-\e^{\frac{2N}{N-1}{H_W}(x)}}.
  \]

  Using  again the Jensen inequality, it follows that
  \begin{align*}
    \eta_x(t)
    &\leq 
    \eta_x(0)
    + \frac{2N}{N-1} \int_0^t \dE_x L H_W (X_s) \,\dd s
    \\
    &\leq 
    \eta_x(0) %
    + \frac{8 \alpha_N}{N^2}  %
    \int_0^t %
    \SBRA{\frac{N}{2}-\dE_x \big[ \e^{\frac{2N}{N-1}{H_W}(X_s)} \big]}\dd s\\
    &\leq
    \eta_x(0) %
    +\frac{8\alpha_N}{N^2}\int_0^t\SBRA{\frac{N}{2}-e^{\eta_x(s)}}\dd s.
  \end{align*}
  Therefore
  \[ 
  \e^{-\eta_x(t)} \geq\e^{-\eta_x(0)-4 \alpha_Nt /N} +\frac{2}{N}(1-\e^{-4
    \alpha_Nt/N}) \geq \min \{ \frac{2}{N}, \e^{-\eta_x(0)} \}
  \]
  by time integration for the first bound and then, for the second bound, by
  writing the obtained expression as the interpolation between $2/N$ and
  $\e^{-\eta_x(0)}$. Dropping the $\e^{-\eta_x(0)-4 \alpha_Nt/N}$ term gives
  the second upper bound in the lemma.
\end{proof}

\begin{proof}[Proof of Theorem \ref{th:Lyapunov}] \ %
  In order to prove that $P^N$ satisfies a Poincaré inequality, we follow the
  approach developed in \cite{MR2386063} based on a Lyapunov function together
  with a local Poincaré inequality (see also the proof of
  \cite[Th.~1.1]{cattiaux-guillin}). This approach amounts to find a positive
  $\cC^2$ function $\phi$ on $D$, a compact set $K\subset D$ and positive
  constants $c,c'$, such that on $D$
  \[
  L \phi\leq -c\phi+c'\IND_K.
  \]
  Such a $\phi$ is called a Lyapunov function. Indeed, for a centered
  $f\in\cF$ this gives
  \[
  \int\!f^2\,\dd P^N
  \leq\int_K\!\frac{c'}{c\phi}f^2\,\dd P^N
  +\int\!-\frac{L\phi}{c\phi}f^2\dd P^N.
  \]
  The first term of the right-hand side can be controlled using a local
  Poincaré inequality, in other words a Poincaré inequality on every ball
  included in $D$, by comparison to the uniform measure. The second one can be
  handled using an integration by parts which is allowed since $f\in\cF$. See
  \cite{MR2386063} and \cite{cattiaux-guillin} for the details.

  For our model $P^N$ we take the $\cC^\infty$ function  
  \[
  \phi=\e^{\ga H}
  \]
  for some $\ga>0$. This function is larger than or equal to $1$ by Lemma
  \ref{le:coercivity}, and the probability measure $P^N$ has a smooth positive
  density on $D$, which provides a local Poincaré constant that may depend on
  $N$ however. 

  Let us check that $\phi$ is a Lyapunov function. To this end, let us show
  that there exist constants $c, c''>0$ and a compact set $K\subset D$ such
  that, on $D$,
  \[
  \frac{L \phi}{\phi}\leq -c+c''\IND_{K}.
  \]
  Indeed, since $\phi$ is positive and bounded on the compact set $K$, this
  gives, on $D$,
  \[
  L \phi\leq -c\phi+c''\sup_{x\in K}\ABS{\phi(x)}\IND_K = -c\phi+c'\IND_K.
  \]
  In order to compute $L \phi/\phi$, we observe that
  \[
  \na \phi=\ga \phi\na H 
  \quad \textrm{and} \quad
  \De \phi =\ga^2 \phi\ABS{\na H}^2+\ga \phi\De H.
  \]
  Therefore, by~\eqref{eq:gen}, 
  \[
  \frac{\beta_N}{\al_N\ga} \frac{L\phi}{\phi}
  =\frac{\De \phi-\be_N\na H\cdot\na \phi}{\ga \phi}
  =\De H+(\ga-\be_N)\ABS{\na H}^2.
  \]
  Now $\Delta H = 4$ on $D$ by~\eqref{eq:deE}. Moreover by
  Lemma~\ref{le:coercivity-nabla}, for $\ga<\be_N$ there exists a compact set
  $K\subset D$ such that
  \[
  (\beta_N - \gamma) \inf_{(x_1,\ldots,x_N)\in K^c}\ABS{\na H}^2(x_1,\ldots,x_N)>5.
  \]
  One can take for instance
  \[
  K=\{x\in(\dR^2)^N:|x|\leq R\text{ and }\min_{i\neq j}\ABS{x_i-x_j}\geq\veps\}
  \]
  for $R>0$ large enough and $\veps>0$ small enough.

  Then $\beta_N /(\alpha_N \gamma) L\phi / \phi \leq -1$ on $D \setminus K$ and the
  Poincaré inequality is proved. Note that we can take $\ga=1$ if $\be_N\geq
  N$.
\end{proof}

\begin{remark}[Poincaré inequality for $P^2$]
  Let us give an alternative direct proof of the Poincaré inequality for the
  probability measure $P^2$. Consider indeed the change of variable
  $(u,v)=((x_1+x_2)/2,(x_1-x_2)/2)$ on $\dR^2 \times \dR^2$, which
  has the advantage to decouple the variables (this miracle is available only
  in the two particle case $N=2$). Letting $\beta = \beta_2$, we get a
  probability density function on $\dR^2\times\dR^2$ proportional to
  \[
  (u,v) \in\dR^2 \times \dR^2\mapsto%
  \e^{-\beta\ABS{u}^2-\beta\ABS{v}^2-+\beta/2 \log\ABS{v}}%
  =\e^{-\beta\ABS{u}^2}  \ABS{v}^{\beta/2} \e^{-\beta \ABS{v}^2}.
  \]
  This probability measure is the tensor product of the Gaussian measure,
  which satisfies a Poincaré inequality, and of the measure $\mu$ with density
  \[
  \frac{\e^{-\beta \Psi(v)}}{Z}
  \quad \text{with} \quad 
  \Psi(v) =  \vert v \vert^2 - \frac{1}{2} \log \vert v \vert.
  \]
  The measure $\mu$ is not log-concave at all (singularity at zero notably)
  but $\Psi(v)$ is a convex function of the norm $r = \vert v \vert$. Hence
  \cite[Th.~1]{MR2083386} ensures that $\mu$ satisfies a Poincaré inequality,
  and then so does our product measure by tensorization.
  
  Note that one can prove Poincaré for $\mu$ by using a Lyapunov function as
  in the proof of Theorem \ref{th:Lyapunov}, instead of
  \cite[Th.~1]{MR2083386}: namely if $L' :=\Delta- \beta \na \Psi \cdot\na$ in
  dimension two and $\phi=\e^{\beta \Psi}/2$, then
  \[
  \frac{L' \phi}{\phi} %
  = \frac{\beta}{2} \De \Psi - \frac{\beta^2}{4} \ABS{\na \Psi}^2,
  \quad
  \na \Psi (v) =2  v-\frac{v}{2 \ABS{v}^2},
  \quad
  \De \Psi=4 
  \]
  for $v\neq0$ (recall that $\log\ABS{v}$ is harmonic in dimension two).
  Therefore
  \[
  \frac{L' \phi}{\phi} %
  = 2 \beta^2 - \beta^2 \Big( \ABS{x}-\frac{1}{4\ABS{x}} \Big)^2 %
  \leq -c + c'\IND_K
  \]
  for the compact set
  \[
  K:=\{x\in\dR^2:r\leq\ABS{x}\leq R\}
  \]
  with $0<r<R$ well chosen.
\end{remark}

\section{Proof of Theorem \ref{th:Bobkov}}
\label{se:Bobkov}

Recall that if $\mu^N$ is the random empirical measure under $P^N$ then for
any continuous and bounded test function $f:\dR^2\to\dR$, using
exchangeability and \eqref{eq:pnphinn},
\begin{align*}
  \dE_{P^N}\int_{\dR^2} \!f(x)\,\mu^N(\dd x)
  &=\int_{(\dR^2)^N}\!\PAR{\frac{1}{N}\sum_{k=1}^Nf(x_k)}\,
  \vphi^{N,N}(x_1,\ldots,x_N)\dd x_1\cdots\dd x_N\\
  &=\int_{\dR^2}\!f(x)\,\vphi^{1,N}(x)\,\dd x
  =\dE_{P^{1,N}}(f)
\end{align*}
where $P^{1,N}$ is the $1$-dimensional marginal of $P^N$. By Theorem
\ref{th:chaos}, as $N\to\infty$, the density $\vphi^{1,N}$ of $P^{1,N}$ tends
to the density of the uniform distribution $\mu_\infty$ on the unit disc of
$\dR^2$. The probability measure $\mu_\infty$ satisfies a Poincaré inequality
for the Euclidean gradient, since for instance it is a Lipschitz contraction
of the standard Gaussian on $\dR^2$. Unfortunately, the convergence of
densities above is not enough to deduce that $P^{1,N}$ satisfies a Poincaré
inequality (uniformly in $N$ or not).

\begin{proof}[Proof of Theorem \ref{th:Bobkov}]
  The idea is to view $P^{1,N}$ as a Boltzmann\,--\,Gibbs measure and to use
  some hidden convexity. Namely, from \eqref{eq:phikn} its density is given on
  $\dR^2$ by
     \begin{equation}\label{phi1N}
  \vphi^{1,N}(x)=\frac{\e^{-N\ABS{x}^2-\psi(\sqrt{N}x)}}{\pi}
  \quad\text{with}\quad
  \psi(x):=-\log\sum_{\ell=0}^{N-1}\frac{\ABS{x}^{2\ell}}{\ell!}.
  \end{equation}
  If we now write $f(x)=\ABS{x}^2+\psi(x)=g(r^2)$ with $r=\ABS{x}$ and
  $g(t)=t-\log\sum_{\ell=0}^{N-1}t^\ell/\ell!$ then
  \[
  \na f(x)=2g'(r^2)x
  \quad\text{and}\quad
  \na^2 f(x)=4g''(r^2)x\otimes x+2g'(r^2)I_2
  \]
  and
  \[
  g'(t)
  =\frac{\frac{t^{N-1}}{(N-1)!}}{\sum_{\ell=0}^{N-1}\frac{t^\ell}{\ell!}}
  \geq0
  \quad\text{and}\quad
  g''(t)
  =\frac{t^{N-2}}{(N-2)!}\frac{
    \PAR{\sum_{\ell=0}^{N-1}\frac{t^\ell}{\ell!}
      -\frac{t}{N-1}\sum_{\ell=0}^{N-2}%
      \frac{t^\ell}{\ell!}}}{\PAR{\sum_{\ell=0}^{N-1}\frac{t^\ell}{\ell!}}^2}
  \geq0.
  \]
  It follows that $f$ is convex (note that its Hessian vanishes at the
  origin), and in other words $P^{1,N}$ is log-concave. Therefore, according
  to a criterion stated in \cite[Th.~1.2]{MR1742893} and essentially due to
  Kannan, Lov\'asz and Simonovits, it suffices to show that the second moment
  of $P^{1,N}$ is uniformly bounded in $N$.
  
  But, using the density~\eqref{phi1N} of $P^{1,N}$, this moment is
  \begin{equation}\label{eq:mom2}
    \int_{\dR^2}\!\!\ABS{x}^2 P^{1,N} (\dd x)
    =\sum_{\ell=0}^{N-1}\frac{N^\ell}{\ell!}%
    \int_0^{\infty}\!\!r^{2(\ell+1)}2r\e^{-Nr^2}\,dr
    =\sum_{\ell=0}^{N-1}\frac{N^\ell}{\ell!}\frac{(\ell+1)!}{N^{\ell+2}}
    =\frac{N+1}{2N} \leq \frac{1}{2} \cdot
  \end{equation}
  
  This concludes the argument thanks to the Bobkov criterion.
\end{proof}
   
  With $\be_N=N^2$ and since $P^{1,N} = \dE \mu^N$, \eqref{eq:mom2} is
  consistent with \eqref{eq:mom2evol} since in this case
  \[
  \lim_{t\to\infty}\dE[{H_V}(X_t)\,|\,X_0=x]%
  =\frac{1}{2}+\frac{1}{N}-\frac{1}{2N}
  =\frac{N+1}{2N}\cdot
  \]
  Note also that, by~\eqref{eq:mom2}, the second moment of $P^{1,N} =
  \dE\mu^N$ tends to $1/2$ as $N\to\infty$; this turns out to be the second
  moment of its weak limit $\mu_\infty$ since
  \[
  \int_{\dR^2} \!\ABS{x}^2\,\mu_\infty(\dd x)
  =\frac{2\pi}{\pi}\int_0^1\!r^3\dd r
  =\frac{1}{2} \cdot
  \]
  Observe finally that a bound on the second moment of $P^{1,N} = \dE\mu^N$
  can be obtained as follows. Let $M$ be a $N\times N$ random matrix with
  i.i.d.\ entries of Gaussian law $\cN(0,\frac{1}{2N}I_2)$ (in other words an
  element of the Complex Ginibre Ensemble); then, by Weyl's inequality
  \cite[Th.~3.3.13]{MR1288752} on the eigenvalues,
  \[
  \int_{\dR^2}\!\ABS{x}^2\,\dE\mu^N(\dd x)
  =\frac{1}{N}\dE\sum_{k=1}^N\ABS{\la_k(M)}^2
  \leq\frac{1}{N}\dE\sum_{k=1}^N\la_k(MM^*)
  =\frac{1}{N}\dE\mathrm{Tr}(MM^*)
  =1.
  \]

\begin{remark}[Poincaré via spherical symmetry]
  The probability measure $P^{1,N}$ is also spherically symmetric, or
  rotationally invariant, as in Bobkov \cite{MR2083386} (see also
  \cite{MR3464047}). Namely, in the notation $f(x)= g( r^2)$ with $r = \vert x
  \vert$ for the ``potential'' of the density of $P^{1,N}$, as in the proof of
  Theorem \ref{th:Bobkov}, let $h( r)=g( r^2).$ Then
  \[
  f(x) = h( r), \quad \na f(x)=h'(r)\frac{x}{\ABS{x}}
  \quad\text{and}\quad
  \na^2f(x)%
  =h''(r)\frac{x\otimes x}{\ABS{x}^2}%
  +h'(r)\frac{\begin{pmatrix}
      x_2^2&-x_1x_2
      \\-x_1x_2&x_1^2
    \end{pmatrix}}{\ABS{x}^{3}} \cdot
  \]
  The matrix on the right-hand side has non-negative trace and null
  determinant, and is thus positive semi-definite (it is the Hessian of the
  norm $x\mapsto\ABS{x}=r$). Moreover
  \[
  h'(t)=2g'(t^2)t,
  \quad
  h''(t)=4g''(t)t^2+2g'(t^2)\geq0.
  \]
  It follows that $P^{1,N}$ is a spherically symmetric probability measure on
  $\dR^2$, and its density is a log-concave function of the norm (and it
  vanishes at the origin). Now according to \cite[Th.~1]{MR2083386}, it
  follows that the probability measure $P^{1,N}$ satisfies a Poincaré
  inequality with a constant which depends only on the second moment, which
  again is bounded in $N$.
\end{remark}
  
\begin{remark}[Logarithmic Sobolev inequality]\label{rk:P1NLSI}
  According to Bobkov's result \cite[Th.~1.3]{MR1742893}, we even get for
  $P^{1,N}$ a logarithmic Sobolev inequality with a uniform constant in $N$
  provided that $P^{1,N}$ has a sub-Gaussian tail uniformly in $N$ (which is
  stronger than the second moment control). This is indeed the case. Namely,
  if $Z\sim P^{1,N},$ then for any real $R\geq0$,
  \[
  \dP(Z\geq R)
  =\int_0^{2\pi}\int_R^{+\infty}\frac{\e^{-Nr^2}}{\pi}
  \sum_{\ell=0}^{N-1}\frac{(Nr^2)^\ell}{\ell!}\,r\dd\te \dd r
  =\frac{1}{N}%
  \int_{NR^2}^{+\infty} \e^{-s}\sum_{\ell=0}^{N-1}\frac{s^\ell}{\ell!}\,\dd s.
  \]
  Moreover
  \[
  \frac{1}{N}\sum_{\ell=0}^{N-1}\frac{s^\ell}{\ell!}
  \leq\frac{s^N}{N!}\leq 2^N\e^{\frac{1}{2}s}
  \]
  for $s\geq N$. Hence, for $R\geq2$,
  \[
  \dP(Z\geq R)\leq\int_{NR^2}^{+\infty} 2^N\e^{-\frac{1}{2}s}\,\dd s
  =2^{N+1}\e^{-\frac{1}{2}NR^2}\leq 4\e^{-\frac{1}{2}R^2}.
  \]
\end{remark}
 
\section{Proof of Theorem~\ref{th:chaos}}\label{se:chaos}

\emph{Proof of the first part of Theorem \ref{th:chaos}.} It is a consequence
of~\eqref{eq:cirlaw} and of the following theorem. Indeed, by Lebesgue's
dominated convergence, ~\eqref{eq:cirlaw} implies that $\dE F(\mu^N)$
tends to $F(\mu_{\infty})$ for every continuous and bounded function $F :
\cP(E) \to \mathbb R.$ In other words, (i) holds in
Theorem~\ref{th:theoA}, whence (ii), which is exactly the first part of
Theorem \ref{th:chaos}.

\begin{theorem}[Characterizations of chaoticity]\label{th:theoA}
  Let $E$ be a Polish space and $\cP(E)$ be the Polish space of Borel
  probability measures on $E$ endowed with the weak convergence topology. Let
  $\mu$ be an element of $\cP(E)$ and let $(P^N)_N$ a sequence of exchangeable
  probability measures on $E^N$. Let us define the random empirical measure
  \[
  \mu^N= \displaystyle \frac{1}{N} \sum_{i=1}^N \delta_{X_i}
  \]
  where $(X_1, \dots, X_N)$ has law $P^N.$ Then the following properties are
  equivalent:
  \begin{itemize}
  \item[(i)] the law of $\mu^N$ converges to $\delta_{\mu}$ weakly in
    $\cP (\cP(E));$
  \item[(ii)] for any fixed $k \leq N$ the $k$-th dimensional marginal
    distribution $P^{k,N}$ of $P^N$ converges weakly in $\cP(E^k)$ to
    the product probability measure $\mu^{\otimes k};$
  \item[(iii)] the $2$-nd dimensional marginal $P^{2,N}$ of $P^N$ converges to
    $\mu^{\otimes 2}$ weakly in $\cP(E^2).$
  \end{itemize}
\end{theorem}

\begin{proof}[Proof of Theorem \ref{th:theoA}]
  Theorem~\ref{th:theoA} is stated for instance in \cite[p.~260]{MR799949},
  \cite[Prop.\ 4.2]{MR1431299} and \cite[Prop.\ 2.2]{MR1108185}, but with a
  sketchy proof that (iii) implies (i). For the reader's convenience, we
  detail this proof when $E$ satisfies the following property : there exists a
  countable subset $D$ of the set $C_b(E)$ of continuous and bounded functions
  $E \to \mathbb R,$ such that for $(\mu_n)_n, \mu$ in $\cP(E)$, it
  holds $\int \phi \, d\mu_n \to \int \phi \, d\mu$ for any $\phi$ in $C_b(E)$
  as soon as it holds for any $\phi$ in $D$. For instance this property holds
  when $E$ is the Euclidean space.

  Since $\cP ( \cP (E))$ is metrizable, it is enough to check
  that for any sequence $(N_k)_k$ there exists a subsequence $(N_{k_j})_j$
  such that the law of $\mu^{N_{k_j}}$ converges to $\delta_{\mu}$. But, by
  expanding the square, exchangeability and (iii),
  \[
  \dE\left(\Big\vert \int_E \phi \, d \mu^{N_k} %
    - \int_E \phi \, d\mu \Big\vert^2\right) \to 0, \qquad k \to + \infty
  \]
  for any $\phi$ in $C_b(E)$ and hence in $D.$ Hence for any such $\phi$ there
  exists a subsequence still denoted $(N_{k_j})_j$ such that $\int \phi \, d
  \mu^{N_{k_j}} \to \int \phi \, d\mu$ almost surely. Now, by a diagonal
  extraction argument, we can build another subsequence $(N_{k_j})_j$ such
  that, almost surely, $\int \phi \, d\mu^{N_{k_j}} \to \int \phi \, d\mu$ for
  any $\phi$ in $D.$ By definition of $D$, this implies that, almost surely,
  $\mu^{N_{k_j}}$ converges to $\mu$ in the metric space $\cP(E)$. It
  follows that the law of $\mu^{N_{k_j}}$ converges to $\delta_{\mu}$ by the
  Lebesgue dominated convergence theorem. Hence (i) since $\cP (
  \cP (E))$ is metrizable.
 \end{proof}

 \emph{Proof of the second part of Theorem \ref{th:chaos}}. We first describe
 the behavior of the one-marginal density function $\vphi^{1,N}$. From
 ~\eqref{eq:phikn} it is given by
  \[
  \vphi^{1,N}(z)=
  \frac{\e^{-N|z|^2}}{\pi}\e_N(N|z|^2),\quad z\in\dC,
  \]
  where $\e_N(w):=\sum_{\ell=0}^{N-1}w^\ell/\ell!$ is the truncated
  exponential series. Then, pointwise in $\mathbb C$,
  \begin{equation}\label{eq:cvphi1N}
    \vphi^{1,N}(z)  \to %
    \frac{1}{\pi} \Big(\mathbf{1}_{|z|<1} + \frac{1}{2}  \mathbf{1}_{|z|=1} \Big), \qquad N \to \infty.
 \end{equation}
 Namely, by rotational invariance, it suffices to consider the case $z = r
 >0$. Next, if $Y_1, \dots, Y_N$ are i.i.d. random variables following the
 Poisson distribution of mean $r^2$, then
\[
\mathrm{e}^{-N r^2}\mathrm{e}_N(Nr^2) =\mathbb{P}(Y_1+\cdots+Y_N<N) =\mathbb{P}\left(\frac{Y_1+\cdots+Y_N}{N}<1\right).
\]
Now, as $N\to\infty$, $\frac{Y_1 + \dots + Y_N}{N} \to r^2$ almost surely by
the law of large numbers, and thus the right-hand side above tends to $0$ if
$r>1$ and to $1$ if $r<1$. In other words
\[
\mathrm{e}^{-N r^2}\mathrm{e}_N(Nr^2) \to \mathbf{1}_{r < 1}
\]
provided $r \neq 1.$ For $r=1$ by the central limit theorem we get
\[
\mathbb{P}\left(\frac{Y_1+\cdots+Y_N}{N}<1\right) = \mathbb{P}\left(\frac{Y_1+\cdots+Y_N - N}{\sqrt{N}}<0\right) \to\frac{1}{2}.
\]

In fact, the convergence in \eqref{eq:cvphi1N} holds uniformly on compact sets
outside the unit circle $\vert z \vert = 1$, as shown in Lemma~\ref{le:exp}
below. It cannot hold uniformly on arbitrary compact sets of $\dC$ since the
pointwise limit is not continuous on the unit circle.

\bigskip

We now turn to the two-marginal density function $\vphi^{2,N}$.
By~\eqref{eq:phikn} it is given by
  \begin{align}
  \vphi^{2,N}(z_1,z_2)
  &=
  \frac{N}{N-1}
  \frac{\e^{-N(|z_1|^2+|z_2|^2)}}{\pi^2}
  \big( \e_N(N|z_1|^2)\e_N(N|z_2|^2)-|\e_N(Nz_1\OL{z}_2)|^2 \big)\nonumber\\
  &=\frac{N}{N-1}\vphi^{1,N}(z_1)\vphi^{1,N}(z_2)
  -
  \frac{N}{N-1}
  \frac{\e^{-N(|z_1|^2+|z_2|^2)}}{\pi^2}
  |\e_N(Nz_1\OL{z}_2)|^2\label{eq:2points}
  \end{align}
    for every $z_1,z_2\in\dC$.
    
  It follows that for any $N\geq2$ and $z_1,z_2\in\dC$,
  \begin{align}\label{defDelta}
  \De_N(z_1,z_2)
  &:= \vphi^{2,N}(z_1,z_2)-\vphi^{1,N}(z_1)\vphi^{1,N}(z_2)\nonumber\\
  &= \frac{1}{N-1}\vphi^{1,N}(z_1)\vphi^{1,N}(z_2) 
  -\frac{N}{N-1}\frac{\e^{-N(|z_1|^2+|z_2|^2)}}{\pi^2}|\e_N(Nz_1\OL{z}_2)|^2.
  \end{align}
  
  \smallskip
  In particular, using
  $\vphi^{2,N}\geq0$ for the lower bound,
  \[
  -\vphi^{1,N}(z_1)\vphi^{1,N}(z_2)
  \leq
  \De_N(z_1,z_2)
  \leq
  \frac{1}{N-1}\vphi^{1,N}(z_1)\vphi^{1,N}(z_2).
  \]
  From this and Lemma \ref{le:exp} we first deduce that for any compact subset
  $K$ of $\{z\in\dC:|z|>1\}$
  \[
  \lim_{N\to\infty}
  \sup_{\substack{z_1\in\dC\\z_2\in K}}|\De_N(z_1,z_2)|
  =
  \lim_{N\to\infty}
  \sup_{\substack{z_1\in K\\z_2\in\dC}}|\De_N(z_1,z_2)|
  =0.
  \]
  
  \smallskip
To conclude the proof of  Theorem  \ref{th:chaos} it remains to show that $\De_N(z_1,z_2)\to0$ as $N\to\infty$ when $z_1$ and $z_2$ are in compact subsets 
of $|z_1| < 1, |z_2| < 1$. In this case
  $|z_1\OL{z}_2|\leq1$, and Lemma \ref{le:exp} gives
  \[
  |\e_N(Nz_1\OL{z}_2)|^2
  \leq 2\e^{2N\Re(z_1\OL{z}_2)}
  +
  2r_N^2(z_1\OL{z}_2).
  \]
  Next, using the elementary identity
  $2\Re(z_1\OL{z}_2)=|z_1|^2+|z_2|^2-|z_1-z_2|^2$, we get
  \begin{equation}\label{enrn}
  \e^{-N(|z_1|^2+|z_2|^2)}|
  \e_N(Nz_1\OL{z}_2)|^2
  \leq
  2\e^{-N|z_1-z_2|^2}
  +2\e^{-N(|z_1|^2+|z_2|^2)}r_N^2(z_1\OL{z}_2).
  \end{equation}
  Since $|z_1\OL{z}_2|\leq1$, the formula for $r_N$ in Lemma \ref{le:exp} gives
  \[
  \e^{-N(|z_1|^2+|z_2|^2)}r^2_N(z_1\OL{z}_2) \leq
  \e^{-N(|z_1|^2+|z_2|^2-2-\log|z_1|^2-\log|z_2|^2)} \frac{(N+1)^2}{2\pi N}.
  \]
  Using \eqref{defDelta}, \eqref{enrn} and the bounds $\varphi^{1,N} \leq
  1/\pi$ and $u-1 - \log u >0$ for $0<u<1$, it follows that $\Delta_N(z_1,
  z_2)$ tends to $0$ as $N\to\infty$ uniformly in $z_1,z_2$ on compact subsets
  of
  \[
  \{(z_1,z_2)\in\dC^2:|z_1|<1,|z_2|<1,z_1\neq z_2\}.
  \]
  This achieves the proof of Theorem \ref{th:chaos}.
 

\begin{lemma}[Exponential series]\label{le:exp}
  Let $\e_N(w):=\sum_{\ell=0}^{N-1}w^\ell/\ell!$ denote the truncated
  exponential series. For every $N\geq1$ and $z\in\dC$,
  \[
  |\e_N(Nz)-\e^{Nz}\mathbf{1}_{|z|\leq1}|\leq r_N(z)
  \]
  where
  \[
  r_N(z):=
  \frac{\e^N}{\sqrt{2\pi N}}|z|^N\PAR{\frac{N+1}{N(1-|z|)+1}
    \mathbf{1}_{|z|\leq1}
    +\frac{N}{N(|z|-1)+1}\mathbf{1}_{|z|>1}}.
  \]
  In particular, for any compact subset
  $K\subset\dC\setminus\{z\in\dC:|z|=1\}$,
  \[
  \lim_{N\to\infty}\sup_{z\in K}
  \ABS{\vphi^{1,N}(z)-\frac{\mathbf{1}_{|z|\leq1}}{\pi}} %
  = \pi^{-1}\lim_{N\to\infty}\sup_{z\in K}
  \ABS{\e^{-N|z|^2}\e_N(N|z|^2)-\mathbf{1}_{|z|\leq1}} =0.
  \]
\end{lemma}

\begin{proof}[Proof of Lemma \ref{le:exp}]
  As in Mehta \cite[Ch.~15]{MR2129906}, for every $N\geq1$, $z\in\dC$, if
  $|z|\leq N$ then
  \[
  \ABS{\e^{z}-\e_N(z)}
  =\ABS{\sum_{\ell=N}^\infty\frac{z^\ell}{\ell!}}
  \leq\frac{|z|^N}{N!}\sum_{\ell=0}^\infty\frac{|z|^\ell}{(N+1)^\ell}
  =\frac{|z|^N}{N!}\frac{N+1}{N+1-|z|},
  \]
  while if $|z|>N$ then
  \[
  |\e_N(z)|
  \leq \sum_{\ell=0}^{N-1}\frac{|z|^\ell}{\ell!}
  \leq \frac{|z|^{N-1}}{(N-1)!}\sum_{\ell=0}^{N-1}\frac{(N-1)^{\ell}}{|z|^\ell}
  \leq\frac{|z|^{N-1}}{(N-1)!}\frac{|z|}{|z|-N+1}.
  \]
  Therefore, for every $N\geq1$ and $z\in\dC$,
  \[
  |\e_N(Nz)-\e^{Nz}\mathbf{1}_{|z|\leq1}|
  \leq
  \frac{N^N}{N!}\PAR{|z|^N\frac{N+1}{N+1-|Nz|}\mathbf{1}_{|z|\leq1}
    +|z|^{N-1}\frac{|Nz|}{|Nz|-N+1}\mathbf{1}_{|z|>1}}.
  \]
  It remains to use the Stirling bound $\sqrt{2\pi N}N^N\leq N!\e^N$ to get
  the first result.
\end{proof}

\subsection*{Acknowledgments} J.F. thanks the hospitality and support of
Université Paris-Dau\-phine via an invited professor position. This work was
partly carried out during a visit to CIRM in Marseille; it is a pleasure for
the authors to thank this institution for its kind hospitality and
participants for discussions on this and related topics, notably Joseph Lehec
and Camille Tardif. The article benefited from a very useful and relevant
anonymous report. The authors acknowledge partial support from the STAB
ANR-12-BS01-0019, Fondecyt 1150570, Basal-Conicyt CMM, Millennium Nucleus
NC120062, and EFI ANR-17-CE40-0030 grants.

\renewcommand{\MR}[1]{}
\bibliographystyle{smfplain}
\bibliography{coulsim}

\end{document}